\documentclass[12pt]{amsart}

\usepackage{amsmath,amssymb,amsthm,a4}

\newcommand{\ie}{\mathrm{i.e.}}

\newtheorem{theorem}{Theorem}
\newtheorem{proposition}{Proposition}
\newtheorem{lemma}{Lemma}
\newtheorem{remark}{Remark}
\newtheorem*{remark*}{Remark}
\newtheorem*{remarks*}{Remarks}

\newcommand{\bbR}{{\mathbb R}}
\newcommand{\bbC}{{\mathbb C}}
\newcommand{\bbN}{{\mathbb N}}
\newcommand{\Hone}{H^{1}({\mathbb R}^N; {\mathbb C})}
\newcommand{\HmOne}{H^{-1}({\mathbb R}^N; {\mathbb C})}
\newcommand{\cL}{{\mathcal L}}
\newcommand{\cM}{{\mathcal M}}
\newcommand{\cN}{{\mathcal N}}
\newcommand{\cE}{{\mathcal E}}
\newcommand{\cT}{{\mathcal T}}
\newcommand{\cR}{{\mathcal R}}

\newcommand{\cC}{{\mathcal C}}

\newcommand{\cF}{{\mathcal F}}
\newcommand{\esssup}{\mathrm{ess}-\mathrm{sup}}
\newcommand{\bx}{{\mathbf x}}
\newcommand{\ba}{{\mathbf a}}
\newcommand{\bv}{{\mathbf v}}
\newcommand{\bk}{{\mathbf k}}
\newcommand{\ta}{{\tilde{\mathbf a}}}
\newcommand{\tv}{{\tilde{\mathbf v}}}
\newcommand{\oa}{{\overline{\mathbf a}}}
\newcommand{\ov}{{\overline{\mathbf v}}}
\newcommand{\bl}{{\mathbf l}}
\renewcommand{\Re}{{\mathrm{Re}}}
\renewcommand{\Im}{{\mathrm{Im}}}

\begin{document}

\title[]{Stochastic acceleration of solitons for the nonlinear Schr\"odinger equation}

\author[]{Walid K. Abou Salem$^{1,*}$ and Catherine Sulem$^{1,**}$}

\address{$^1$Department of Mathematics, University of Toronto, Toronto, Ontario, Canada M5S 2E4 \\ E-mail: walid@math.utoronto.ca; sulem@math.utoronto.ca.} 
\address{$^*$ Partially supported by NSERC grant NA 7901.}
\address{$^{**}$ Partially supported by NSERC grant 46179-05.}
\maketitle

\begin{abstract}
The effective dynamics of solitons for the generalized nonlinear Schr\"odinger equation in a  random potential is rigorously studied.  It is shown that when the external potential varies slowly in space compared to the size of the soliton, the dynamics of the center of the soliton is almost surely described by Hamilton's equations for a classical particle in the  random potential, plus error terms due to radiation damping. Furthermore, a limit theorem for the dynamics of the center of mass of the soliton in the weak-coupling  and space-adiabatic limit is proven in two and higher dimensions: Under certain mixing hypotheses for the potential, the momentum of the center of mass of the soliton converges in law to a diffusion process on a sphere of constant momentum. Moreover, in three and higher dimensions, the trajectory of the center of mass of the soliton converges to a spatial Brownian motion. 
\end{abstract}

\section{Introduction}
\subsection{Overview of earlier results and heuristic discussion}

In the last few years, there has been substantial progress in  rigorously understanding solitary wave dynamics for the nonlinear Schr\"odinger equation in slowly varying potentials (or in the presence of small rough perturbations), see \cite{BJ1}-\cite{HZ2}. The basic picture is that in the space-adiabatic limit, or in the presence of small perturbations, the long-time dynamics of the center of the soliton is described by Hamilton's (or Newton's) equations in an effective potential that corresponds to the restriction of the external potential to the soliton manifold, plus error terms due to radiation damping.  

In the above cited work, a soliton behaves like a classical point particle over certain scales. The main 
question that we address is whether this analogy between solitons and point particles over a certain 
temporal and spatial scale still holds when the soliton is moving in the field of a random potential. We 
show that the answer to the above question is affirmative. Furthermore, we have strong enough explicit control over the  
long-time dynamics of the soliton in the  random potential that allows us to look at the limiting 
dynamics over large-distances and long-times. 
In a certain weak-coupling /space-adiabatic limit, and under some mixing hypotheses for the random 
potential, we show that the long-time, large-distance  behavior of the soliton is described by momentum 
diffusion in $N\ge 2$ dimensions: The soliton center of mass undergoes Brownian motion on the energy sphere 
of constant momentum. 
This is analogous to the long-time/large-distance behavior of the a classical particle in a random 
potential, \cite{KP} - \cite{DGL}. Moreover, in dimensions $N\ge 3,$ the long-time limit of a momentum 
diffusion is a spatial Brownian motion, \cite{KR2}.  We also show this spatial diffusive behavior for the soliton center of mass. We note that the weak-coupling and space-adiabatic limit are taken {\it simultaneously}. This is more difficult than taking the semi-classical limit first, and then the weak-coupling limit. 

This diffusive motion of the center of mass has been observed numerically in \cite{B-SDM,DiMen} for the NLS equation with power nonlinearity and white  multiplicative noise; see also \cite{BK1} for a discussion of a   concrete experiment where such a dynamics may be observed. For the derivation of the NLS equation in the mean-field limit of interacting bosons in a random potential, we refer the reader to \cite{A-S3}.

We note that a somewhat related problem arises in the semi-classical limit of the dynamics of a quantum particle in a random field. It is shown by Erd\"os, Salmhofer and Yau in \cite{ESY} that the semi-classical/weak-coupling limit of the dynamics of a quantum particle in a  random potential displays a  diffusive behavior of the energy density of solutions of the {\it linear} Schr\"odinger equation. 
Another related problem is that of localization for the nonlinear Schr\"odinger equation, which has been studied recently by Bourgain and Wang in \cite{BW1}.

\subsection{Description of the problem}\label{sec:Description}

In what follows, we consider the probability triple $(\Omega, {\mathcal F}, {\mathbb P}),$ such that the probability space $\Omega$ has a generic point $\omega$ and is endowed with measure $\mu.$ For a measurable and integrable function $f$ on $\Omega,$ we define the expectation value of $f$ as ${\mathbb E}(f):= \int f(\omega) \mu(d\omega).$

In this paper, we study the long-time dynamics of solitary wave solutions for the generalized nonlinear Schr\"odinger equation in a  random potential
\begin{equation}
\label{eq:NLSE}
i\partial_t \psi  (\bx,t) = (-\Delta + \lambda V_h(\bx,t;\omega))\psi  (\bx,t) -f(\psi  (\bx,t)),
\end{equation}
where $\omega\in \Omega,$ $\bx\in {\mathbb R}^N$ denotes a point in the configuration space, $t\in {\mathbb R}$ is time, $\partial_t = \frac{\partial}{\partial t}, \Delta = \sum_{j=1}^N \frac{\partial^2}{\partial{x_j}^2}$ the $N$-dimensional Laplacian, $V_h$ is the  random potential, which is a measurable real function on ${\mathbb R}^N\times {\mathbb R}\times\Omega\rightarrow {\mathbb R}$ satisfying
\begin{equation*}
V_h(\bx,t;\omega)\equiv V(h\bx,t;\omega),\; h\in(0,1],
\end{equation*}
$\lambda\in [0,1]$ is a coupling constant, 
and the nonlinearity $f$ is a mapping on complex Sobolev spaces, 
\begin{equation*}
f:H^1({\mathbb R}^N; {\mathbb C}) \rightarrow H^{-1}({\mathbb R}^N; {\mathbb C}), 
\end{equation*}
with $f(0)=0,$ and $\overline{f(\psi)}=f(\overline{\psi}),$ where $\overline{\cdot}$ denotes complex conjugation.  

Typical focusing nonlinearities are local ones
\begin{equation}
\label{eq:LocalNL}
f(\psi)= |\psi|^s \psi,   \ \ 0<s<\frac{4}{N}.
\end{equation}

When $\lambda=0,$ the NLS equation (\ref{eq:NLSE}) with nonlinearity as given above admits solitary wave solutions, which are stable  spherically symmetric and positive solutions
\begin{equation}
 \eta_{\sigma}(\bx ,t):= e^{i(\frac{1}{2}\bv\cdot(\bx - \ba)+\gamma)}\eta_\mu(\bx - \ba), \label{eq:Sol}
\end{equation}
where $\sigma:=(\ba,\bv,\gamma,\mu)$, 
$\ba=\bv t+\ba_0$, $\gamma=\mu t +  \frac{\|\bv\|^2}{4}t+ \gamma_{0}$, with $\gamma_0 \in 
[0,2\pi)$, $\ba_0,\bv\in\bbR^N,$ 
$\mu\in\mathbb{R}^{+}$, 
and $\eta_\mu$ is a positive solution of the 
nonlinear eigenvalue problem
\begin{equation}
(-\Delta + \mu )\eta_\mu -f(\eta_\mu)=0.
\end{equation}
The solution (\ref{eq:Sol}) stands for a soliton  traveling  with velocity $\bv,$ center of mass position $\ba,$ and phase $\frac{1}{2}\bv\cdot(\bx-\ba)+\gamma.$ The {\it size} of the soliton is $\propto \mu^{-1/2},$ in the sense that $\eta_\mu \sim e^{-\sqrt{\mu}\|\bx\|}$ as $\|\bx\|\rightarrow\infty,$ see for example \cite{Ca1} and \cite{SS1}. 

We consider in this paper an external potential whose realizations are $\mathbb{P}$-a.s. varying slowly in space compared to the size of the soliton, $i.e.,$ 
\begin{equation*}
\esssup_{\omega \in\Omega, \; t\in \bbR, \; \bx\in\bbR^N} \frac{|\nabla V(\bx,t;\omega)|}{\sqrt{\mu}}\ll 1,
\end{equation*}
which corresponds to the space-adiabatic limit ($h\ll 1$) if we set the size of the soliton to $O(1).$ Furthermore, we assume that the realizations of the external potential satisfy 
\begin{equation}
\label{eq:Pot1}
V_h \in W^{1,\infty}(\bbR, C^2(\bbR^N)\cap W^{2,\infty}(\bbR^N)) \; {\mathbb P} \mathrm{- a.s.},
\end{equation}
i.e., there exists $\overline{\Omega}\subset \Omega$ with $\mu(\overline{\Omega})=1,$ such that $$V_h, \partial_t V_h \in L^\infty (\bbR, C^2(\bbR^N)\cap W^{2,\infty}(\bbR^N)) \; \mathrm{for} \; \omega\in\overline{\Omega}.$$

We are interested in the dynamics of a single soliton in the external potential. We assume that there exists $\sigma_0 =(\ba_0,\bv_0,\gamma_0,\mu_0)$ such that 
\begin{equation}
\label{eq:InitialCond}
\|\psi(t=0) - \eta_{\sigma_0}\|_{H^1} \le C_0 h,
\end{equation}
for some positive constant $C_0,$ where $h$ appears in (\ref{eq:NLSE}).
Theorem \ref{th:Main1} below is a rigorous result in the space-adiabatic limit for the long-time dynamics of the soliton for the special class of nonlinearities and potential discussed above. A more general result, Theorem \ref{th:Main3}, is stated in Section \ref{sec:MainResults} after listing general assumptions, and is proven in Section  \ref{sec:ProofMain3}.

\begin{theorem}\label{th:Main1}
Consider the NLS equation (\ref{eq:NLSE}) with nonlinearity $f$ given by (\ref{eq:LocalNL}) and potential $V_h$ satisfying (\ref{eq:Pot1}) , and suppose the initial condition satisfies (\ref{eq:InitialCond}). Then there exists $h_0>0$ and positive constants $C$ and $\overline{C},$ such that, for all $h\in (0,h_0)$ and for any fixed $\epsilon\in (0,1),$  we have that 
\begin{equation*}
\sup_{t\in [0,\overline{C}\epsilon |\log h| /\lambda h )}{\mathbb E}[\|\psi - \eta_{\sigma(t)}\|_{H^1}] \le C h^{1-\frac{\epsilon}{2}},
\end{equation*}
uniformly in $\lambda\in (h^{1-\epsilon},1],$ and  the parameters $\sigma(t) = (\ba(t),\bv(t),\gamma(t),\mu(t))$ satisfy the system of effective  equations 
\begin{align*}
&\partial_t \ba = \bv + O(h^{2-\epsilon})\\
&\partial_t \bv = -2\lambda\nabla V_h(\ba,t;\omega) + O(h^{2-\epsilon})\\
&\partial_t\gamma = \mu + \frac{1}{4}\|\bv\|^2 - V_h(\ba,t;\omega)  + O(h^{2-\epsilon})\\
&\partial_t\mu = O(h^{2-\epsilon}), 
\end{align*}
with initial condition satisfying
\begin{equation*}
\|\ba(0)-\ba_0\|, \|\bv(0)-\bv_0\|,|\gamma(0)-\gamma_0|, |\mu(0)-\mu_0| = O(h).
\end{equation*}

\end{theorem}

We now discuss a limit theorem for the dynamics of the center of mass of the soliton moving in a {\it time-independent} and {\it strongly mixing} random potential.  Suppose that  
\begin{equation}
\label{eq:TimeIndepPot1}
V_h(\bx,t;\omega) \equiv \overline{V}(h\bx; \omega),  \ \ \omega\in\Omega .
\end{equation}

Performing the simple scaling 
\begin{align*}
&\overline{\ba}:= h\ba \\
&\overline{\bv}:= \bv\\
&\overline{t}:= ht,
\end{align*}
it follows from Theorem \ref{th:Main1} that, for $\overline{t}\in [0,\overline{C}\epsilon |\log h|/\lambda ),$ 
\begin{align}
\label{eq:RescaledCMDyn}
&\partial_{\overline{t}}\overline{\ba}=\ov + O(h^{2-\epsilon})\\
&\partial_{\overline{t}}\overline{\bv}=-2\lambda \nabla \overline{V}(\overline{\ba}; \omega) + O(h^{1-\epsilon}),
\end{align}
with initial condition 
\begin{equation}
\label{eq:RescaledInitCond2}
\|\oa(0) - h\ba_0\| = O(h^2), \; \|\ov(0)-\bv_0\|=O(h),
\end{equation}
which are, up to small remainder terms, the equations of motion of a classical particle in a random potential. This motivates introducing an auxiliary stochastic process $(\tilde{\ba}(\overline{t}),\tilde{\bv}(\overline{t}))_{\overline{t}\ge 0}$ that is given by 
\begin{align}
&\partial_{\overline{t}} \tilde{\ba} = \tilde{\bv} \label{eq:Res2a1}\\
&\partial_{\overline{t}} \tilde{\bv} = -2\lambda \nabla \overline{V}(\tilde{\ba};\omega),\label{eq:Res2v1}
\end{align}
with initial condition 
\begin{equation}
\label{eq:Res2InitCond1}
\tilde{\ba}(0)= 0, \ \ \tilde{\bv}(0)=\bv_0.
\end{equation}
This process corresponds to a classical particle with Hamiltonian
$$H_{cl}(\tilde{\ba},\tilde{\bv}) = \frac{\|\tilde{\bv}\|^2}{2} + 2\lambda \overline{V}(\tilde{\ba}).$$ The dynamics of a classical particle in a weak, strongly mixing and homogeneous random potential has been studied extensively in the literature, see \cite{KP}-\cite{KR2}.

We assume that $\overline{V}: \bbR^N\times\Omega \rightarrow \bbR$ is measurable such that it is strictly stationary in space, i.e., the laws of $(\overline{V}(\bx_1+\bx_0)\cdots \overline{V}(\bx_n+\bx_0))$ and $(\overline{V}(\bx_1)\cdots \overline{V}(\bx_n))$ are the same for all $n\in {\mathbb N}$ and $\bx_0\in\bbR^N.$ We also assume that 
\begin{equation}
\label{eq:TimeIndepPot2}
{\mathbb E}(\overline{V})=0,
\end{equation}
and that the realizations of the potential 
\begin{equation}
\label{eq:TimeIndepPot3}
\overline{V} \in C^2(\bbR^N)\cap W^{2,\infty}(\bbR^N), \ \ {\mathbb P}-\mathrm{a.s.}.
\end{equation}

%%%%%%%%%%%%%%%%%%%%%%%%%%%%%%%%%%%%%%%%%%%%

Furthermore, we assume that $\overline{V}$ is strongly mixing in the uniform sense. For any $r>0,$ let ${\mathcal C}_r^i$ and ${\mathcal C}_r^e$ be the $\sigma$-algebras generated by the random variables $\overline{V}(\bx),$ $\bx\in B_r:= \{\bx\in\bbR^N, \ \ \|\bx\| <r\}$ and $\bx\in B_r^c$ respectively, where $B_r^c$ is the complement of the open ball $B_r$ in $\bbR^N.$ We define the uniform mixing coefficient between the $\sigma$-algebras   ${\mathcal C}_r^i$ and ${\mathcal C}_r^e$  as 
\begin{equation}
\varphi(\rho) = \sup\{|{\mathbb P}(A)-{\mathbb P}(B|A)|, \ \ r>0, \; A\in {\mathcal C}_r^i, \; B\in {\mathcal C}_{r+\rho}^e\}, \rho>0.
\end{equation}
We assume that $\varphi (\rho)$ decays faster than any power of $\rho,$ i.e., for all $p>0,$
\begin{equation}
\label{eq:TimeIndepPot5}
\sup_{\rho\ge 0} \rho^p \varphi(\rho) <\infty.
\end{equation}

We introduce the two-point spatial correlation function corresponding to the random potential
\begin{equation}
\label{eq:CorrFunct}
R(\ta):= 4{\mathbb E}(\overline{V}(\ta)\overline{V}(0)),
\end{equation}
with Fourier transform $\widehat{R}(\tv).$ We assume that 
\begin{equation}
R\in C^\infty (\bbR^N),
\end{equation}
such that 
\begin{equation}
\label{eq:SmoothCorr}
\widehat{R} \; \mathrm{does}\; \mathrm{not}\; \mathrm{vanish}\; \mathrm{identically} \; \mathrm{on}\; \mathrm{any}\; H_{{\mathbf p}}:=\{\tv\in\bbR^N, \; \tv\cdot {\mathbf p} = 0\}, {\mathbf p}\in\bbR^N.
\end{equation}
 
In order to observe diffusive behavior of the soliton dynamics for all ``macroscopic'' times 
$1/\lambda^2 < \overline{C}\epsilon |\log h|/\lambda,$ as the coupling constant $\lambda\rightarrow 0$ and $h\rightarrow 0,$ we need that $|\log h| \lambda\rightarrow \infty$ as $h\rightarrow 0.$ This is satisfied, for example, if $\lambda=\frac{1}{|\log h|^{1-\alpha}},$ for some $\alpha\in (0,1).$ 

We now discuss the limit process as $\lambda\rightarrow 0$ and $h\rightarrow 0.$ Let 
$(\underline{\bv}(\overline{t}))_{\overline{t}\ge 0}$ be a diffusion process starting at $\underline{\bv}(\overline{t}=0)=\bv_0,$ where $\bv_0$ appears in (\ref{eq:InitialCond}), with generator ${\mathcal L}$ whose action on $u\in C_0^\infty (\bbR^d\backslash \{0\})$ is given by 
\begin{equation}
\label{eq:Diff1}
{\mathcal L} u(\underline{\bv}) = \sum_{i,j=1}^N \partial_{\underline{v}_i}(D_{ij}(\underline{\bv})\partial_{\underline{v}_j} u(\underline{\bv})).
\end{equation}
Here $D_{ij}$ is the diffusion matrix given by 
\begin{equation}
\label{eq:DiffMatrix}
D_{ij}(\bk):= -\frac{1}{2\|\bk\|} \int_{-\infty}^\infty \partial_{x_i}\partial_{x_j} R(s \frac{\bk}{\|\bk\|}) ds, \ \ \bk\in\bbR^N.
\end{equation}
We have the following theorem on the convergence of the momentum of the center of mass of the soliton to a Brownian motion on a sphere of constant momentum. A stronger result stated for more general nonlinearities, Theorem \ref{th:Main4} in Section  \ref{sec:MainResults}, will be proven in Section  \ref{sec:ProofMain4}. 

\begin{theorem}\label{th:Main2}
Consider the NLS equation (\ref{eq:NLSE}) in $N\ge 2$ dimensions with nonlinearity $f$ given by (\ref{eq:LocalNL}) and potential $V$ satisfying (\ref{eq:TimeIndepPot1})-(\ref{eq:SmoothCorr}). Suppose the initial condition satisfies (\ref{eq:InitialCond}) with $\|\bv_0\|\ne 0,$ and that there exists $\tilde{\alpha}>0$ such that the coupling constant $\lambda\rightarrow 0$ as $h\rightarrow 0$ with $|\log h| \lambda^{3/2+\tilde{\alpha}}\rightarrow \infty.$   Then, for any fixed and finite $T >0,$ the stochastic process $(\lambda^2 \overline{\ba}(\overline{t}/\lambda^2), \overline{\bv}(\overline{t}/\lambda^2))_{\overline{t}\in(0,T )},$ satisfying (\ref{eq:RescaledCMDyn}) - (\ref{eq:RescaledInitCond2}), converges weakly (in law), as $\lambda, h\rightarrow 0,$ to $(\int_0^{\overline{t}} \underline{\bv}(s)ds, \underline{\bv}(\overline{t}))_{\overline{t}\in(0,T )},$ where $\underline{v}$ satisfies (\ref{eq:Diff1}).
\end{theorem} 

Actually, in dimensions $N\ge 3,$ we have a stronger result about the convergence of the trajectory of the soliton's center of mass to a spatial Brownian motion. 

Let $\phi^\lambda(\bx,t,\bk)$ satisfy the Liouville equation
\begin{equation}
\label{eq:Liouville1}
\partial_t \phi^\lambda = \partial_{\overline{t}} \oa|_{\oa=\bx,\ov=\bk }\cdot \nabla_{\bx} \phi^\lambda + \partial_{\overline{t}} {\ov}|_{\oa=\bx,\ov=\bk}\cdot \nabla_{\bk} \phi^\lambda .
\end{equation}
with initial condition $\phi^\lambda(\bx,0,\bk)= \phi_0(\lambda^{2+\beta}\bx,\bk), \beta >0,$ where $\phi_0$ is compactly supported, twice differentiable in $\bx\in\bbR^N,$ and four times differentiable in $\bk\in\bbR^N,$ such that its support is contained in the shell
\begin{equation*}
{\mathcal A}(M):= \{(\bx,\bk)\in \bbR^{2N}, \ \ 1/M < \|\bk\|< M\},
\end{equation*}
for some $M>1.$
Let $u(\bx,t,\bk)$ be the solution of the spatial diffusion 
\begin{equation}
\label{eq:SpaceDiffusion}
\partial_t u = \sum_{i,j} d_{ij}(\|\bk\|) \partial_{x_i}\partial_{x_j}u,
\end{equation}
with initial condition 
\begin{equation*}
u(\bx,0,\bk) = \frac{1}{\Gamma_{N-1}} \int_{S^{N-1}} \phi_0(\bx,\|\bk\| \bl) d\Sigma(\bl),
\end{equation*}
where $\Gamma_{N-1}$ is the area of a unit sphere $S^{N-1}$ in $N$-dimensions, and $d\Sigma(\bl)$ is the measure on the unit sphere $S^{N-1}.$ Here, the coefficients $d_{ij}$ are given by 
\begin{equation*}
d_{ij}(\|\bk\|) = \frac{1}{\Gamma_{N-1}} \int_{S^{N-1}} \|\bk\|  l_i \chi_j (\|\bk\| \bl )d\Sigma(\bl),
\end{equation*}
where $\chi_j$ are the mean-zero solutions of 
\begin{equation*}
\sum_{i,j=1}^N \partial_{k_i}(D_{ij}(\bk)\partial_{k_j}\chi_j) = -\|\bk\| \hat{k}_j.
\end{equation*}
We have the following theorem.

\begin{theorem}\label{th:Main2Extended}
Consider the NLS equation (\ref{eq:NLSE}) in $N\ge 3$ dimensions with nonlinearity $f$ given by (\ref{eq:LocalNL}) and potential $V$ satisfying (\ref{eq:TimeIndepPot1})-(\ref{eq:SmoothCorr}). Suppose the initial condition satisfies (\ref{eq:InitialCond}) with $\|\bv_0\|\ne 0,$ and that there exists $\tilde{\alpha}>0$ such that $\lambda\rightarrow 0$ as $h\rightarrow 0$ with $|\log h| \lambda^{1+\tilde{\alpha}}\rightarrow \infty.$ Then there exists $\tilde{\beta} \in (0,\tilde{\alpha}/2)$  such that, for all $0< \beta<\tilde{\beta},$  any fixed $0<t_0< T  < \infty ,$ and all compact sets $K\subset{\mathcal A}(M),$ we have 
\begin{equation*}
\lim_{\lambda, h\rightarrow  0}\sup_{(t,\bx,\bk)\in [t_0,T ]\times K} |{\mathbb E}[\phi^\lambda(\bx/\lambda^{2+\beta},t/\lambda^{2+2\beta},\bk)]- u(\bx,t,\bk) | =0,
\end{equation*}
where $\phi^\lambda$ satisfies (\ref{eq:Liouville1}) and $u$ satisfies (\ref{eq:SpaceDiffusion}).
\end{theorem}

The same result holds for more general nonlinearities, see Theorem \ref{th:Main5}, Section  \ref{sec:MainResults}. We note that the conditions on the convergence of $\lambda\rightarrow 0$ as $h\rightarrow 0$ in Theorem \ref{th:Main2Extended} are optimal, in the sense that they are the ones needed for the time scale over which the soliton behaves like a point particle, $O(|\log h|/\lambda),$ is larger than the macroscopic time scale over which spatial diffusion is observed, $O(\lambda^{-2-2\beta}), \; \beta>0.$

Our analysis of the effective dynamics of the soliton in the external potential is based on an extension 
of \cite{A-S1} to the case where the potential has a random character and on the derivation of
explicit estimates involving both the spatial variation of the external potential $h$ and the coupling constant of the potential $\lambda$; see also \cite{FJGS1}-\cite{HZ2}. This analysis relies on three main ingredients. First, using a skew-orthogonal (or Lyapunov-Schmidt) decomposition property, we decompose, almost surely, the solution of the NLS equation with initial condition close to a soliton configuration into a path belonging to the soliton manifold and a part describing a fluctuation skew-orthogonal to the manifold. The dynamics on the soliton manifold is obtained by the skew-orthogonal projection of the Hamiltonian flow generated by the NLS equation in a small tubular neighbourhood of the soliton manifold onto the latter. As for the  fluctuation, we control its $H^1$-norm almost surely using an approximate Lyapunov functional. To study the limiting dynamics in the weak-coupling/space-adiabatic limit, we rely on explicit control of the difference between the effective dynamics of the soliton and a classical particle in the random potential and  on the results of \cite{KP},\cite{KR1} and \cite{KR2} on the stochastic acceleration of a classical particle in a  random potential. 

The organization of this paper is as follows. In Section \ref{sec:MainResults}, we list our general assumptions and state our main 
results for general nonlinearities: Theorem \ref{th:Main3} on the long-time dynamics of the soliton in the random potential, and 
Theorems \ref{th:Main4} and \ref{th:Main5} on the limiting diffusive dynamics in the 
weak-coupling/space-adiabatic limit. In Section \ref{sec:NLSeq}, we recall basic properties of the 
NLS equation, which we will use in the analysis. We then prove Theorem 
\ref{th:Main3} in Section  \ref{sec:ProofMain3}. In Section  \ref{sec:ProofMain4} we prove Theorems \ref{th:Main4} and \ref{th:Main5} by  showing first that the limiting effective dynamics of the soliton converges almost surely to that of a classical particle in a random potential and  then applying the results of \cite{KP}, \cite{KR1} and \cite{KR2} on the motion of a classical particle governed by a weakly random Hamiltonian flow. For  sake of completeness, we discuss the well-posedness of the generalized NLS with random interactions and potential in Appendix A. We also prove some technical statements in Appendix B.

\subsection{Notation}
\begin{itemize}

\item Given $\bx\in \bbR^N,$ we denote $x \equiv \|\bx\|:=\sqrt{\sum_{i=1}^N x_i^2}$ and $\hat{\bx}:= \frac{\bx}{\|\bx\|}.$ For $\bx,{\mathbf y}\in\bbR^N,$ we denote their inner product by $\bx\cdot {\mathbf y}:=\sum_{i=1}^N x_iy_i. $

\item $L^p(I)$  denotes the standard Lebesgue space, $1\le p\le \infty,$ with norm
\begin{align*}
\|f\|_{L^p} &= (\int_I dx~ |f(x)|^p)^{\frac{1}{p}}, \; f\in L^p(I),  p<\infty , 
\\ \|f\|_{L^\infty} &= \esssup(|f|) , \; f\in L^\infty (I).
\end{align*}
We also define
\begin{equation*}
\|f\|_{L^p(I,L^q(J))}:= \|\, \|f\|_{L^q(J)}\, \|_{L^p(I)}.
\end{equation*}
For $1\le p\le \infty,$ $p'$ is the conjugate of $p, \ie , 1/p+1/p' = 1.$
We denote by $\langle \cdot ,\cdot \rangle $ the scalar product in $L^2(\bbR^N),$ 
\begin{equation*}
\langle u, v\rangle = \Re\int_{\bbR^N} u\overline{\bv}, \; u,v\in L^2(\bbR^N),
\end{equation*}
as well as its extension by duality to ${\mathsf Y}\times {\mathsf Y}',$ where ${\mathsf Y}$ and ${\mathsf Y}'$ are complete metric spaces such that ${\mathsf Y}\hookrightarrow L^2 \hookrightarrow {\mathsf Y}',$ with dense embedding

\item Given the multi-index $\overline{\alpha}= (\alpha_1,\cdots ,\alpha_N) \in \bbN^N,$ we denote $|\overline{\alpha}| = \sum_{i=1}^N \alpha_i.$ Furthermore, $\partial_x^{\overline{\alpha}} := \partial_{x_1}^{\alpha_1}\cdots \partial_{x_N}^{\alpha_N}.$

\item
For $1\le p\le\infty$ and $s\in\bbN,$ the (complex) Sobolev space is given by
\begin{equation*}
W^{s,p}(\bbR^N) := \{u\in {\mathcal S}'(\bbR^N): \partial_x^{\overline{\alpha}} u\in L^p(\bbR^N), |\overline{\alpha}|\le s\}, 
\end{equation*}
where ${\mathcal S}'(\bbR^N)$ is the space of tempered distributions. We equip $W^{s,p}$ with the norm 
\begin{equation*}
\|u\|_{W^{s,p}} = \sum_{\overline{\alpha}, |\overline{\alpha}|\le s} \|\partial_x^{\overline{\alpha}} u \|_{L^p},
\end{equation*}
which makes it a Banach space. Moreover, $W^{-s,p'}$ is the dual of $W^{s,p}.$  We use the shorthand $W^{s,2}=H^s.$

\item Given $f$ and $g$ real functions on $\bbR^N,$ we denote their convolution by $\star,$
\begin{equation*}
f\star g(\bx):=\int d {\mathbf y}~ f(\bx-{\mathbf y}) g({\mathbf y})  .
\end{equation*}

\item We say that a real function $f\in C_b^{1,1,2}(\bbR^N\times [0,+\infty)\times \bbR^N\backslash \{0\})$ if it satisfies
\begin{equation*}
\sup_{|\overline{\alpha}|\le 1,\beta \le 1, |\overline{\gamma}|\le 2}\; \sup_{(\bx,t,\bk)\in \bbR^N\times [0,+\infty)\times \bbR^N \backslash \{0\}} |\partial_\bx^{\overline{\alpha}} \partial_t^\beta \partial_\bk^{\overline{\gamma}}f(\bx,t,\bk)|<\infty.
\end{equation*}

\end{itemize}

\section{General nonlinearities}\label{sec:MainResults}

\subsection{The model}

We now list our general assumptions and discuss models where they are satisfied.

\begin{itemize}
\item[(A1)] {\it Nonlinearity.} 
The nonlinearity $f=f_1+\cdots+f_k$ such that each 
\begin{equation*}
f_j\in C^2(\Hone,\HmOne), \; j=1,\cdots ,k.
\end{equation*} 
We assume that there exists 
\begin{equation*}
r_j\in [2,\frac{2N}{N-2}), \; ([2,\infty] \ \ \mathrm{if}\ \ N=1),
\end{equation*} 
such that $\forall M>0, \exists$ a finite constant $C_j(M)$ such that 
\begin{equation*}
\| f_j(u)-f_j(v)\|_{L^{r_j'}}\le C_j(M)\|u-v\|_{L^{r_j}},
\end{equation*}
$\forall u,v\in \Hone, \|u\|_{H^1}+\|v\|_{H^1}\le M.$ Furthermore, 
\begin{equation*}
\Im f_j(u)\overline{u}=0
\end{equation*}
almost everywhere on $\bbR^N, \forall u\in \Hone.$

Let $F:= \sum_{j=1}^k F_j,$ where $F_j$ satisfies $f_j=F'_j,$ and the prime stands for the Fr\'echet derivative. We have that 
\begin{equation*}
\exists F_j \in C^3(\Hone,\bbR).
\end{equation*}   
For every $M>0,$ there exists a positive constant $C(M)$ and $\tilde{\epsilon} \in (0,1)$ such that
\begin{equation*}
F(u)\le \frac{1-\tilde{\epsilon}}{2}\|u\|_{H^1} + C(M), \forall u\in H^1(\bbR^N)
\end{equation*}
such that $\|u\|_{L^2}\le M.$ Furthermore, 
\begin{align*}
&\sup_{\|u\|_{H^1}\le M} \|F''(u)\|_{{\mathcal B}(H^1,H^{-1})} < \infty\\
&\sup_{\|u\|_{H^1}\le M} \|F'''(u)\|_{H^1\rightarrow {\mathcal B}(H^1,H^{-1})} < \infty,
\end{align*}
where ${\mathcal B}$ denotes the space of bounded operators.

\item[(A2)] {\it Symmetries.} The nonlinearity $F$ satisfies $F(T\cdot)=F(\cdot),$ where $T$ is a translation 
\begin{equation*}
T_\ba^{tr}: u(\bx)\rightarrow u(\bx-\ba), \; \ba\in \bbR^N,
\end{equation*}
or a rotation 
\begin{equation*}
T_R^r: u(\bx)\rightarrow u(R^{-1}\bx), \; R\in SO(N),
\end{equation*}
or a gauge transformation
\begin{equation*}
T_\gamma^g : u(\bx)\rightarrow e^{i\gamma}u(\bx), \; \gamma \in [0,2\pi),
\end{equation*}
or a boost 
\begin{equation*}
T_v^b: u(\bx)\rightarrow e^{\frac{i}{2}\bv\cdot \bx} u(\bx), \; \bv\in \bbR^N
\end{equation*}
or a complex conjugation
\begin{equation*}
T^c: u(\bx)\rightarrow \overline{u}(\bx).
\end{equation*}

\item[(A3)]{\it Solitary Wave.} $\exists I\subset \bbR$ such that $\forall \mu\in I,$ the nonlinear eigenvalue problem 
\begin{equation}
\label{eq:NLEV}
(-\Delta +\mu)\eta_\mu -f(\eta_\mu)=0
\end{equation}
has a positive, spherically symmetric solution $\eta_\mu\in L^2(\bbR^N)\cap C^2(\bbR^N),$ such that 
\begin{equation*}
\|\|\bx\|^3\eta_\mu\|_{L^2}+\|\|\bx\|^2\|\nabla \eta_\mu\|\|_{L^2} + \|\|\bx\|^2\partial_\mu \eta_\mu\|_{L^2}<\infty, \forall \mu \in I.
\end{equation*}

\item[(A4)]{\it Orbital Stability.} The solution $\eta_\mu$ appearing in assumption (A4) satisfies
\begin{equation*}
\partial_\mu m(\mu) >0, \forall \mu\in I,
\end{equation*} 
where $m(\mu): = \frac{1}{2} \int d\bx ~ \eta_\mu^2$ is the ``mass'' of the soliton.

\item[(A5)] {\it Null Space Condition.} We define
\begin{equation*}
\cL_\mu := -\Delta +\mu -f'(\eta_\mu),
\end{equation*}
which is the Fr\'{e}chet derivative of the map $\psi\rightarrow (-\Delta+\mu)\psi -f(\psi)$ evaluated at $\eta_\mu.$ For all $\mu\in I,$ the null space
\begin{equation*}
\cN(\cL_\mu) = span\{i\eta_\mu, \partial_{x_j}\eta_\mu, j=1,\cdots, N\}.
\end{equation*}

\item[(B1)] {\it External Random Potential.}
The external potential is a real measurable function on $\bbR^N\times\bbR\times\Omega\rightarrow \bbR$ such that
\begin{equation*}
V_h(\bx,t;\omega) \equiv V(h\bx, t; \omega), \ \ h\in (0,1],
\end{equation*}
and it realizations 
\begin{equation*}
V(\bx,t;\omega) \in W^{1,\infty}(\bbR, C^2(\bbR^N)\cap W^{2,\infty}(\bbR^N)), \ \ {\mathbb P}-\mathrm{a.s.}.
\end{equation*}

\item[(B2)] {\it Time-Independent, Homogeneous and Strongly Mixing Random Potential.} The external potential is time-independent, 
\begin{equation*}
V(\ba,t ;\omega) \equiv \overline{V}(\ba ; \omega),
\end{equation*}
and $\overline{V}$ is a homogeneous random potential that is strongly mixing in the uniform sense, satisfying conditions (\ref{eq:TimeIndepPot1})-(\ref{eq:SmoothCorr}), Subsection  \ref{sec:Description}.

\end{itemize}

\begin{remark} Assumptions (A1) and (B1) are sufficient to establish almost surely the global well-posedness of the NLS equation (\ref{eq:NLSE}) in $H^1$ (see Appendix A). Moreover, (A1) implies 
\begin{align*}
&|F(u+v)-F(u)-\langle F'(u), v\rangle | \le C(M) \|v\|_{H^1}^2 \\
&|F(u+v)-F(u)-\langle F'(u), v\rangle- \frac{1}{2}\langle F''(u)v,v\rangle| \le C(M)\|v\|_{H^1}^3 \\
&\|F'(u+v)-F'(u)-F''(u)v\|_{H^{-1}}\le C(M)\|v\|_{H^1}^2 ,
\end{align*}
for any $u\in H^2(\bbR^N)$ and $v\in H^1(\bbR^N)$ such that $\|u\|_{H^1}+\|v\|_{H^1}\le M.$ 
\end{remark}

\begin{remark}
We now discuss typical nonlinearities where the various assumptions are satisfied. Assumptions (A1) and (A2) are satisfied, for example, if 
\begin{equation*}
F(u)= \frac{1}{2}\int d\bx ~ G(|u|^2) + W\star|u|^2,
\end{equation*}
where $G(r)=\int_0^r ds g(s),$ such that $g\in C^2(\bbR^+)$ with $g(s)\le C(1+s^\alpha),\alpha\in [0,\frac{2}{N}),$ $|\partial_s^k g(s)|\le C(1+s^{q-k}), k=0,1,2, q\in [0,\frac{2}{N-2}), \; {if} \; N\ge 3,$ $q\in [0,\infty), \; {if} \;  N=1,2.$ Furthermore, $W\in L^p+L^\infty, p>\frac{N}{2}, p\ge 1,$ such that $\max(0,W)\in L^r(\bbR^N)+L^\infty(\bbR^N), r> \frac{N}{2},\ge 1, \; {if} \; N\ge 2,$ and $W$ is spherically symmetric, see \cite{Ca1}. Furthermore, (A3) is satisfied for local nonlinearities if 
\begin{align*}
   -\infty &< \lim_{s\rightarrow 0} g(s) < \mu \\
   -\infty &\le\lim_{s\rightarrow\infty}s^{-\alpha}g(s) \le C ,  
\end{align*}   
where $0<\alpha < 2/(N-2)$, when $N>2$ and
$\alpha \in(0,\infty)$ if  $N=1,2,$ such that  
\begin{equation*}
\exists \zeta>0, \ \text{with}\ 
        \int^\zeta_0ds  g(s) >\mu\zeta,
\end{equation*}
see \cite{BL1,BLP1,Ca1}. Assumption (A3) is also satisfied for nonlocal nonlinearities if, in addition to the above, 
\begin{equation*}
    W\in L^{q}_{\mathrm{loc}},\ q\ge\frac{N}{2}, \; {and} \; 
W\rightarrow 0\ \text{as}\ \|\bx\|\rightarrow \infty;
\end{equation*}
see \cite{Ca1,GV2,FTY1,FTY2}. Assumption (A4) implies orbital stability of the solitary wave solution, \cite{GSS1}. It is satisfied, in particular, for local nonlinearities $$f(\psi)=|\psi|^{s}\psi, s<\frac{4}{N}.$$ 
Assumption (A5) is satisfied for local nonlinearities if 
\begin{equation*}
g'(s)+g''(s)s^2 >0,
\end{equation*}
or if $N=1,$ \cite{FJGS1, We1}. 
\end{remark}

\subsection{Statement of the main theorems for general nonlinearities}

We now state the main results of this paper.

\begin{theorem}\label{th:Main3}
Consider the NLS equation (\ref{eq:NLSE}) with initial condition satisfying (\ref{eq:InitialCond}). Suppose that (A1)-(A5) and (B1) hold. Then the result of Theorem \ref{th:Main1} hold.

\end{theorem}

We now discuss convergence of the dynamics of the center of mass of the soliton to Brownian motion in the weak-coupling/space-adiabatic limit.

\begin{theorem}\label{th:Main4}
Consider the NLS equation (\ref{eq:NLSE}) in dimensions $N\ge 2$ and with initial condition satisfying (\ref{eq:InitialCond}) and $\|\bv_0\|\ne 0.$ Assume hypotheses (A1)-(A5), (B1) and (B2) hold. Suppose in addition that there exists $\tilde{\alpha}>0$ such that  $\lambda\rightarrow 0$ as $h\rightarrow 0$ with $|\log h| \lambda^{3/2+\tilde{\alpha}}\rightarrow \infty.$  Then the result of Theorem \ref{th:Main2} hold.
\end{theorem}

We also have the following stronger result for the convergence of the trajectory of the soliton to spatial Brownian motion in  the weak-coupling/space-adiabatic limit in dimensions $N\ge 3.$

\begin{theorem}\label{th:Main5}
Consider the NLS equation (\ref{eq:NLSE}) in dimensions $N\ge 3$ and with initial condition satisfying (\ref{eq:InitialCond}) and $\|\bv_0\|\ne 0.$ Assume hypotheses (A1)-(A5), (B1) and (B2) hold. Suppose in addition that  there exists $\tilde{\alpha}>0$ such that $\lambda\rightarrow 0$ as $h\rightarrow 0$ with $|\log h| \lambda^{1+\tilde{\alpha}}\rightarrow \infty.$ Then the result of Theorem \ref{th:Main2Extended} hold.
\end{theorem}

We prove Theorem \ref{th:Main3} in Section  \ref{sec:ProofMain3}, and Theorems \ref{th:Main4} and \ref{th:Main5} in Section  \ref{sec:ProofMain4}. Before doing so, we recall some useful properties of the NLS equation.

\section{Mathematical preliminary}\label{sec:NLSeq}

In this section, we recall some basic properties of the nonlinear Schr\"odinger equation (\ref{eq:NLSE}), see for example \cite{FJGS1}. 

\subsection{Hamiltonian structure}
 
The phase space for the NLS equation (\ref{eq:NLSE}) is chosen to be $\Hone.$
The space $\Hone$ has a real inner product (Riemannian metric) 
\begin{equation}
\langle u,v \rangle :=\Re \int d\bx ~ u\overline{v}  
\end{equation}
for $u,v\in H^1(\bbR^N;\bbC).$ The tangent space at an element $\psi\in H^1$ is $\cT_\psi H^1=H^1.$ On $\Hone,$ one can define a symplectic 2-form 
\begin{equation}
\Xi(u,v):= \Im \int  d\bx ~ u\overline{\bv} = \langle u,iv\rangle.
\end{equation}
The Hamiltonian functional corresponding to the NLS equation (\ref{eq:NLSE}) is
\begin{equation}
\label{eq:NLHamiltonian}
H_{\lambda}(\psi):=\frac{1}{2}\int d\bx  ~ |\nabla\psi|^2  +\frac{\lambda}{2}\int  V_h |\psi|^2 - F(\psi).
\end{equation}

Using the correspondence 
\begin{align*}
\Hone &\longleftrightarrow H^1(\bbR^N;\bbR) \oplus H^1(\bbR^N;\bbR)\\
\psi &\longleftrightarrow (\Re\psi \; \Im\psi)\\
i^{-1}&\longleftrightarrow J,
\end{align*}
where $J:=\begin{pmatrix} 0 & 1 \\ -1 & 0 \end{pmatrix}$ is the complex structure on $H^1(\bbR^N;\bbR^{2}),$ the NLS equation can be written as 
\begin{equation*}
\partial_t \psi   = J H_{\lambda}' (\psi  ).
\end{equation*}

We note that since the Hamiltonian functional $H_{\lambda}$ defined in (\ref{eq:NLHamiltonian}) is nonautonomous, energy is {\it not } conserved. For $\psi\in H^1$ satisfying (\ref{eq:NLSE}), we have that 
$$\partial_t H_{\lambda}(\psi)= \frac{\lambda}{2}\int d\bx ~ (\partial_t V_h) |\psi|^2,\ \ {\mathbb P}-\mathrm{a.s.},$$
see Appendix B for a proof of this statement. Yet, $H_{\lambda}$ is almost surely invariant under global gauge transformations, 
\begin{equation*}
H_{\lambda}(e^{i\gamma}\psi)=H_{\lambda}(\psi), \; {\mathbb P}-\mathrm{a.s.}, 
\end{equation*}
and the associated conserved quantity is the ``charge'' 
\begin{equation*}
N(\psi):= \frac{1}{2}\int d\bx~ |\psi|^2.
\end{equation*}

The assumption $\partial_\mu m(\mu)>0,$ where $m(\mu)$ is defined in (A4), implies that $\eta_\mu$ appearing in assumption (A3) is a local minimizer of $H_{\lambda=0}(\psi)$ restricted to the balls ${\mathcal B}_m:= \{\psi\in H^1 : N(\psi)=m\},$ for $m>0;$ see \cite{GSS1}. They are critical points of the action functional 
\begin{equation}
\label{eq:EnergyFunctional}
\cE_\mu (\psi) := \frac{1}{2} \int d\bx~ (|\nabla\psi|^2 +\mu|\psi|^2)-F(\psi),
\end{equation}
where $\mu=\mu(m)$ is a Lagrange multiplier.

\subsection{Soliton Manifold}

Let $\eta_\mu$ be an (orbitally stable) soliton solution of (\ref{eq:NLEV}), and define  
the soliton manifold as 
\begin{equation*}
\cM_s:= \{\eta_\sigma:= T_{\ba\bv\gamma}\eta_\mu , \; \sigma= (\ba,\bv,\gamma,\mu) \in \bbR^N\times\bbR^N\times [0,2\pi)\times I \},
\end{equation*}
where $I$ appears in assumption (A3), and the combined transformation $T_{\ba\bv\gamma}$ is given by
\begin{equation*}
\psi_{\ba\bv\gamma}:= T_{\ba\bv\gamma}\psi = e^{i(\frac{1}{2}\bv\cdot (\bx-\ba)+\gamma)}\psi(\bx-\ba),
\end{equation*}
where $\bv,\ba\in\bbR^N$ and $\gamma\in [0,2\pi).$ Note that if the nonlinearity $f$ satisfies $f'(0)=0,$ then $I\subset \bbR^+.$ 

The tangent space to $\cM_s$ at $\eta_\mu\in \cM_s$ is given by 
\begin{equation*}
\cT_{\eta_\mu}\cM_s = span\{ E_t,E_g,E_b,E_s\},
\end{equation*}
where
\begin{align*}
E_t &:= \nabla_\ba T_\ba^{tr}\eta_\mu|_{\ba=0}=-\nabla\eta_\mu\\
E_g &:= \partial_\gamma T_\gamma^g \eta_\mu|_{\gamma=0} = i\eta_\mu\\
E_b &:= 2\nabla_\bv T_\bv^{b}\eta_\mu|_{\bv=0}=i\bx\eta_\mu\\
E_s &:= \partial_\mu\eta_\mu.
\end{align*}
In the following, we denote by 
\begin{align}
&e_j   := -\partial_{ x_j  }, \;  j=1,\cdots ,N,\nonumber \\ 
&e_{j+N} := i  x_j  , \; j=1,\cdots ,N, \nonumber\\
&e_{2N+1} := i, \nonumber\\
&e_{2N+2} := \partial_\mu, \label{eq:TangentBases}
\end{align}
which, when acting on $\eta_\sigma\in\cM_s,$ generate the basis vectors $\{e_  \alpha   \eta_\sigma\}_{  \alpha  =1}^{2N+2}$ of $\cT_{\eta_\sigma}\cM_s.$

The soliton manifold $\cM_s$ inherits a symplectic structure from $(H^1,\Xi).$ For $\sigma = (\ba,\bv,\gamma,\mu) \in \bbR^N\times\bbR^N\times [0,2\pi)\times I,$ the matrices
\begin{equation*}
\Xi_\sigma :=  P_\sigma J^{-1} P_\sigma 
\end{equation*}
where $P_\sigma$ is the $L^2$-orthogonal projection onto $\cT_{\eta_\sigma}\cM_s,$ define the induced symplectic structure on $\cM_s.$  Explicitly, we have
\begin{align}
\Xi_{\sigma}|_{\cT_{\eta_\sigma}\cM_s} :&= \{\langle e_  \alpha  \eta_\sigma, i e_\beta   \eta_\sigma \rangle\}_{1\le \alpha,\beta\le 2N+2}\nonumber \\
&= 
\left(
\begin{matrix}
0 & -m(\mu) {\mathbf 1}_{N\times N} & 0 & -\frac{1}{2}\bv m'(\mu) \\
m(\mu) {\mathbf 1}_{N\times N} & 0 & 0 & \ba m'(\mu) \\
0 & 0 & 0 & m'(\mu) \\
\frac{1}{2}\bv^T m'(\mu) & -\ba^Tm'(\mu) & -m'(\mu) & 0
\end{matrix}
\right),
\label{eq:Metric}
\end{align}
where ${\mathbf 1}_{N\times N}$ is the $N\times N$ identity matrix, and $(\cdot)^T$ stands for the transpose of a vector in $\bbR^N.$ One can easily show that if $\partial_\mu m(\mu)>0,$ then $\Xi_\sigma$ is invertible.

\subsection{Group structure}
The anti-selfadjoint operators $\{e_  \alpha  \}_{\alpha=1,\cdots,2N+1}$ defined in (\ref{eq:TangentBases}) form the generators of the Lie algebra ${\mathsf g}$ corresponding to the Heisenberg group ${\mathsf H}^{2N+1},$ where the latter is given by 
\begin{equation*}
(\ba,\bv,\gamma)\cdot(\ba',\bv',\gamma')=(\ba'',\bv'',\gamma''),
\end{equation*}
with $\ba''=\ba+\ba',$ $\bv''=\bv+\bv',$ and $\gamma''=\gamma'+\gamma+\frac{1}{2}\bv\cdot \ba'.$
Elements of ${\mathsf g}$ satisfy the commutation relations
\begin{equation}
[e_i,e_{j+N}]=-e_{2N+1}\delta_{ij}, \; i,j= 1,\cdots,N , \label{eq:CommutationRelations}
\end{equation}
and the rest of the commutators are zero.

\subsection{Zero modes}
The solitary wave solutions transform covariantly under translations and gauge transformations, i.e.,
\begin{equation*}
{\mathcal E}_\mu' (T_\ba^{tr} T_\gamma^g \eta_\mu)=0
\end{equation*}
for all $\ba\in \bbR^N$ and $\gamma\in [0,2\pi).$ Here, the prime stands for the Fr\'echet derivative. 

There are zero modes of the {\it Hessian},   
\begin{equation}
\label{eq:Hessian}
\cL_\mu:= -\Delta +\mu -f'(\eta_\mu),
\end{equation} 
associated to these symmetries. One can show that 
\begin{equation*}
i\cL_\mu: \cT_{\eta_\mu}\cM_s \rightarrow \cT_{\eta_\mu}\cM_s
\end{equation*}
with $(i\cL_\mu)^2 X = 0,$ for any vector $X\in \cT_{\eta_\mu}\cM_s.$ 

To see this, differentiate ${\mathcal E}_\mu' (T_\ba^{tr} \eta_\mu)=0$ with respect to $\ba$ and set $\ba$ to zero, which gives 
\begin{equation}
\label{eq:ZeroModes1}
\cE''(\eta_\mu)\nabla_\ba\eta_\mu(\bx-\ba) |_{\ba=0}= \cL_\mu E_t = 0.  
\end{equation}
Similarly, differentiating ${\mathcal E}_\mu'(T_\gamma^g \eta_\mu)=0$ with respect to $\gamma$ and setting $\gamma$ to zero gives
\begin{equation}
\cL_\mu E_g = 0.
\end{equation}
Using (\ref{eq:NLEV}), we have 
\begin{equation}
\cL_\mu E_b = iE_t,
\end{equation}
and
\begin{equation}
\label{eq:ZeroModes2}
\cL_\mu E_s = iE_g. 
\end{equation}

\subsection{Skew-Orthogonal Decomposition}\label{sec:SOD}

Consider the manifold $\cM_s' = \{\eta_\sigma, \ \ \sigma\in \Sigma_0\}, \ \ \Sigma_0 = \bbR^N\times \bbR^N \times [0,2\pi) \times I_0,$ where $I_0\subset I\backslash \partial I$ is bounded. We define the $\delta$ neighbourhood of $\cM_s'$ in $H^1$ as 
\begin{equation}
\label{eq:SODNbr}
U_\delta := \{ \psi\in H^1 , \ \ \inf_{\sigma\in \Sigma_0} \|\psi - \eta_\sigma\|_{H^1}\le \delta\}.
\end{equation}
Then, for $\delta$ small enough and for all $\psi\in U_\delta,$ there exists a unique $\sigma(\psi)\in C^1(U_\delta,\Sigma)$ such that 
\begin{equation*}
\Xi(\psi-\eta_{\sigma(\psi)}, X)= \langle \psi - \eta_{\sigma(\psi)}, i X\rangle =0,
\end{equation*}
for all $X\in\cT_{\eta_{\sigma(\psi)}}\cM_s.$ For a proof of this statement, we refer the reader to \cite{FJGS1}.

\begin{remark}\label{rem:SODPrime} 
The group element $T_{\ba\bv\gamma}\in {\mathsf H}^{2N+1}$ is given by 
\begin{equation*}
T_{\ba\bv\gamma} = e^{-\ba\cdot \partial_\bx} e^{i\frac{\bv\cdot \bx}{2}} e^{i\gamma} .
\end{equation*}
It follows from (\ref{eq:CommutationRelations}) that $T^{-1}_{\ba\bv\gamma} YT_{\ba\bv\gamma} \in {\mathsf g}$ if $Y\in {\mathsf g}.$ Furthermore, it follows from translational invariance that $\Xi(T_{\ba\bv\gamma} u, T_{\ba\bv\gamma} v)=\Xi(u,v),$ for $u,v\in L^2.$ Therefore,  
\begin{equation}
\label{eq:SO}
\Xi(w, Y \eta_{\sigma}) = \Xi ( w', Y' \eta_{\sigma'} ) = 0, \nonumber 
\end{equation}
$\forall Y\in {\mathsf g},$ where $Y'=T^{-1}_{\ba\bv\gamma} Y T_{\ba\bv\gamma}\in {\mathsf g},$ $w'= T^{-1}_{\ba\bv\gamma}w,$ and $\eta_{\sigma'}=T^{-1}_{\ba\bv\gamma}\eta_\sigma.$
\end{remark}

%%%%%%%%%%%%%%%%%%%%%%%%%%%%%%%%%%%%%%%%%%%%%%%%%%%%%%%%%%%%%%%%%%%%%%%
%LONG-TIME DYNAMICS
%%%%%%%%%%%%%%%%%%%%%%%%%%%%%%%%%%%%%%%%%%%%%%%%%%%%%%%%%%%%%%%%%%%%%%%

\section{Long-time dynamics in a random potential}\label{sec:ProofMain3}

In this section, we study the long-time dynamics of the center of mass of the soliton, and we prove Theorem \ref{th:Main3}.  The main ingredient of our analysis is an extension of the analysis of solitary wave dynamics in time-dependent  potentials, \cite{A-S1}, to the case of random potentials; see also \cite{FJGS1}-\cite{HZ1}, \cite{A-S2} and \cite{HZ2} for the time-independent case. Furthermore, we improve slightly the error bounds in \cite{A-S1} by obtaining estimates that depend both on the coupling constant $\lambda$ and the scale of the spatial variation of the external potential, $h.$ 

The weak solution of (\ref{eq:NLSE}) with initial condition $\phi$ is given by the Duhamel formula
\begin{equation}
\label{eq:DuhamelEq1}
\psi  (t) = U(t,0)\phi -i\lambda \int_0^t ds~ U(t,s)  V_h(\bx,t)\psi (s)  +i\int_0^t~ds U(t,s)f(s,\psi  (s)),
\end{equation}
where $U(t,0)= e^{i\Delta t}$ is the unitary operator corresponding to free time evolution. One can show that it is unique, and that it is ${\mathbb P}$-a.s. in $H^1$; see Appendix A for a proof. It follows from (\ref{eq:DuhamelEq1}) that $\psi  $ is $\omega$-measurable. We have the following proposition.

\begin{proposition}\label{pr:Main}
Consider (\ref{eq:NLSE}) with initial condition satisfying (\ref{eq:InitialCond}).
Suppose assumptions (A1)-(A5) and (B1) hold. Then there exists $h_0>0$ and positive constants $C$ and $\overline{C},$ such that, for all $h\in (0,h_0),$ any fixed $\epsilon\in (0,1),$  and all times $t\in [0,\overline{C} \epsilon |\log  h|/\lambda h),$ we have that 
\begin{equation*}
\psi(\bx,t)= \eta_{\sigma(t)}(\bx)+w(\bx,t), \ \ \omega\in\overline{\Omega},
\end{equation*}
with $$\sup_{\omega\in\overline{\Omega}}\; \sup_{t\in [0,\overline{C} \epsilon |\log  h|/\lambda h)}\|w\|_{H^1}\le C h^{1-\frac{\epsilon}{2}},$$ uniformly in $\lambda\in (h^{1-\epsilon},1],$ and the parameters $\ba,\bv,\gamma$ and $\mu$ satisfy the effective differential equations, for $\omega\in\overline{\Omega}$ and $t\in [0,\overline{C}\epsilon |\log h|/\lambda h),$
\begin{align}
\partial_t \ba &= \bv+ O(h^{2-\epsilon}) \label{eq:EffectiveA}\\
\partial_t \bv &= -2 \lambda \nabla  V_h(\ba,t) + O(h^{2-\epsilon})\\
\partial_t \gamma &= \mu -\lambda V_h(\ba,t)+\frac{1}{4}v^2 + O(h^{2-\epsilon})\\
\partial_t \mu &= O(h^{2-\epsilon}),\label{eq:EffectiveMu}
\end{align}
with $\|\ba(0)-\ba_0\|,\|\bv(0)-\bv_0\|, |\gamma(0)-\gamma_0|, |\mu(0)-\mu_0| = O(h).$
\end{proposition}

The proof of Proposition \ref{pr:Main} is given in Subsection \ref{sec:ProofPrMain}. We first find the reparametrized equations of motion corresponding to the restriction of the Hamiltonian flow generated by the NLS equation to the soliton manifold. 

\begin{remark}
In the special case of local nonlinearities in dimension $N=1,$ and under the strong assumption on the initial condition,
$$\psi(t=0)=\eta_{\sigma_0},$$
one obtains a slightly better estimate on the $H^1$ norm of the fluctuation $w$ in Proposition \ref{pr:Main}, 
$$\sup_{\omega\in\overline{\Omega}}\; \sup_{t\in [0,\overline{C} \epsilon |\log  h|/\lambda h)}\|w\|_{H^1}\le C h^{2-\frac{\epsilon}{2}},$$ see, for example, \cite{HZ2} for a discussion of  the cubic  one
dimensional NLS 
equation  with a slowly varying time-independent potential.  Here, we choose to work with general nonlinearities in all dimensions under weaker hypotheses for the initial condition.
\end{remark}

\subsection{Reparametrized equations of motion}\label{sec:RepEqMotion}

Suppose $\psi  $ satisfies the initial value problem (\ref{eq:NLSE}) such that, ${\mathbb P}$-almost surely, $\psi   \in U_\delta \subset H^1,$ where $U_\delta$ is defined in (\ref{eq:SODNbr}), Subsection \ref{sec:SOD}. By the skew-orthogonal decomposition, there exists a unique 
$$\sigma=\sigma(\psi)=(\ba,\bv,\gamma,\mu)\in \Sigma = \bbR^N\times \bbR^N\times [0,2\pi)\times I$$ and $w  '\in H^1$ such that 
\begin{equation*}
\psi   = \eta_\sigma + w  ' , \ \ \forall \omega\in\overline{\Omega}, 
\end{equation*}
with 
\begin{equation*}
w'    \perp J^{-1} \cT_{\eta_\sigma}\cM_s .
\end{equation*}  
For $\omega\in\overline{\Omega},$ let
\begin{equation}
\label{eq:CMSol}
u  := T^{-1}_{\ba\bv\gamma}\psi   = \eta_\mu + w  , 
\end{equation}
where $w  =T^{-1}_{\ba\bv\gamma}w  '.$ 

We introduce the coefficients  
\begin{equation*}
c_j := \partial_t a_j - v_j , \ \ c_{N+j}:=-\frac{1}{2} \partial_t v_j - \lambda \partial_{x_j} V_h(\ba,t) , \ \ j=1,\cdots ,N ,
\end{equation*}
\begin{equation}
\label{eq:Coefficients}
c_{2N+1} := \mu -\frac{1}{4}v^2 + \frac{1}{2}\partial_t \ba \cdot \bv - \lambda V_h(\ba,t)-\partial_t\gamma , \ \ c_{2N+2}:=-\partial_t \mu.
\end{equation}
We denote by 
\begin{equation}
\label{eq:Alpha}
|c|:= \sup_{i\in \{1,\cdots ,2N+2\}}|c_i|,
\end{equation}
and 
$$C (|c|,w,h):= \sup_{\omega\in\overline{\Omega}}\{|c|\|w  \|_{H^1}+\lambda h^2 + \|w  \|_{H^1}^2\}.$$ We have the following lemma.

\begin{lemma}\label{lm:RepEqMotion}
Consider the NLS equation (\ref{eq:NLSE}) with initial condition (\ref{eq:InitialCond}). Suppose assumptions (A1)-(A5) and (B1) hold, and that for $\omega\in\overline{\Omega},$ $\psi  (t)\in U_\delta$ for $t\in [0,T],$ where $\delta$ and $T$ are independent of $\omega\in\overline{\Omega}.$ Then the parameter $\sigma=(\ba,\bv,\gamma,\mu),$ as given above, satisfy 
\begin{align*}
\partial_t a_j &= v_j + O(C  (|c|,w  ,h))\\
\partial_t v_j &= -2 \partial_{x_j} \lambda V_h(\ba,t;\omega) + O(C  (|c| ,w  ,h)) \\
\partial_t \gamma &= \mu -\frac{1}{4} \|\bv\|^2 +\frac{1}{2}\partial_t \ba \cdot \bv -\lambda V_h(\ba,t;\omega) + O(C  (|c|, w  ,h)) \\
\partial_t \mu &= O(C  (|c|,w  ,h)) ,
\end{align*}
for $t\in (0,T), j=1,\cdots N.$
\end{lemma}

\begin{proof}

We first find the equations of motion in the center of mass reference frame.
For $\omega\in\overline{\Omega},$ we differentiate 
$$u (\bx,t) =T^{-1}_{\ba\bv\gamma}\psi (\bx,t)   = e^{-\frac{i}{2}(\bv\cdot \bx + \gamma)} \psi  (\bx+\ba)$$ 
with respect to $t,$ and use the fact that 
\begin{align*}
&e^{-\frac{i}{2}(\bv\cdot \bx +\gamma)}\Delta \psi  (\bx+\ba) = \Delta u(\bx,t)   + i \bv\cdot \nabla u(\bx,t)   -\frac{\|\bv\|^2}{4} \\
&e^{-\frac{i}{2} (\bv\cdot \bx + \gamma)}f(\psi  (\bx+\ba,t)) = f(u(\bx,t)).
\end{align*}
We have that
\begin{equation}
\label{eq:CenterMassDyn}
\partial_t u   = -i ((-\Delta +\mu)u   -f(u  )) + \sum_{j=1}^{2N+1}c_j e_j u   -i \cR_V u  
\end{equation}
where 
\begin{equation*}
\cR_V := \lambda V_h(\bx+\ba,t;\omega) - \lambda V_h(\ba,t;\omega)- \lambda \nabla V_h(\ba,t;\omega) \cdot \bx,
\end{equation*}
and $e_j$ and $c_j$ are as defined in (\ref{eq:TangentBases}) and (\ref{eq:Coefficients}) respectively. 
Equivalently, we have that  
\begin{equation*}
\partial_t u   = -i \cE_\mu'(u  ) + \sum_{j=1}^{2N+1} c_j e_j u   -i\cR_V u  ,
\end{equation*}
where $\cE_\mu$ appears in (\ref{eq:EnergyFunctional}). 

We now use (\ref{eq:CenterMassDyn}) and the skew-orthogonal decomposition to find equations for the parameters $\sigma=(\ba,\bv,\gamma,\mu)$ and $w,$ for $\omega\in\overline{\Omega}.$ Since
\begin{equation*}
\cE_\mu'(\eta_\mu)=0,
\end{equation*}
we have that 
\begin{equation*}
\cE_\mu' (u  )= \cL_\mu w   + N_\mu (w  ), 
\end{equation*}
where $\cL_\mu = (-\Delta +\mu -f'(\eta_\mu)) = \cE_\mu''(\eta_\mu)$ and $N_\mu (w  )= f(\eta_\mu+w  ) - f(\eta_\mu) -f'(\eta_\mu)w  .$ Substituting this back in (\ref{eq:CenterMassDyn}) gives
\begin{equation*}
\partial_t w   = (-i\cL_\mu + \sum_{j=1}^{2N+1} c_j e_j -i\cR_V)w   + N_\mu (w  ) + \sum_{j=1}^{2N+2}  c_j e_j \eta_\mu -i \cR_V \eta_\mu .
\end{equation*}
Now, we know that $\langle X,i w  \rangle = 0$ for all $X\in \cT_\eta \cM_s,$ $\omega\in\overline{\Omega},$ see Remark \ref{rem:SODPrime}. It follows that 
\begin{equation*}
\partial_t \langle X,i w  \rangle = \partial_t \mu \ \ \langle \partial_\mu X, iw  \rangle + \langle X, i\partial_t w   \rangle = 0.
\end{equation*}  
Therefore, for $\omega\in\overline{\Omega},$ we have that 
\begin{align*}
&\partial_t \mu\ \   \langle \partial_\mu X,i w  \rangle = \langle iX, \partial_t w  \rangle \\
&= -\langle iX, i \cL_\mu w  \rangle - \langle iX, i \cR_V(\eta_\mu + w  ) \rangle +\langle iX, N_\mu (w  )\rangle  +\langle iX, \sum_{j=1}^{2N+1}  c_j e_j w  \rangle \\ &+\langle iX ,\sum_{j=1}^{2N+2} c_j e_j \eta_\mu \rangle .
\end{align*}
It follows from (\ref{eq:ZeroModes1})-(\ref{eq:ZeroModes2}) that $\langle iX, i\cL_\mu w  \rangle =0.$ Together with 
\begin{equation*}
[e_j, i]= 0, \ \ e_j^* = - e_j, \ \ j=1,\cdots ,2N+2,
\end{equation*}
we have that
\begin{equation}
\label{eq:IntDyn}
\sum_{j=1}^{2N+2} c_j \langle   e_j X,i w  \rangle = -\langle X, \cR_V(w  +\eta_\mu) + N_\mu (w  ) \rangle + \sum_{j=1}^{2N+2}c_j \langle X,i  e_j\eta_\mu \rangle ,
\end{equation}
for $\omega\in\overline{\Omega}.$
Now, choosing $X=E_k,$ where $E_k, k\in \{1,\cdots , 2N+2\},$ is a basis vector of $\cT_{\eta_\mu} \cM_s$ gives 
\begin{equation*}
\sum_{j=1}^{2N+2} (\Xi)_{kj} c_j = \langle E_k, N_\mu (w  )+\cR_V(w  +\eta_\mu)\rangle + \sum_{j=1}^{2N+2}c_j \langle i e_j E_k, w   \rangle , \ \ \omega\in\overline{\Omega},
\end{equation*} 
where $\Xi$ appears in (\ref{eq:Metric}).
Replacing the definition of $c_j,$ appearing in (\ref{eq:Coefficients}), in (\ref{eq:IntDyn}), and using the fact that 
\begin{equation*}
\langle x_j \eta_\mu, i\cR_V \eta_\mu \rangle = 0 , \ \ \eta_\mu (x) = \eta_\mu (\|x\|), \ \ \langle \eta_\mu, i\cR_V \eta_\mu\rangle =0,
\end{equation*}
gives, for $\omega\in\overline{\Omega},$
\begin{equation}
\label{eq:a}
\partial_t a_k = v_k + \frac{1}{m(\mu)}(\langle i x_k\eta_\mu,  N_\mu (w)-i\cR_Vw \rangle + \sum_{j=1}^{2N+2} c_j \langle e_j x_k \eta_\mu, w   \rangle) , 
\end{equation}
\begin{equation}
\label{eq:v}
\partial_t v_k = -2 \partial_{x_k} \lambda V_h(\ba,t;\omega) + \frac{2}{m(\mu)} (\langle \partial_{x_k}\eta_\mu, N_\mu (w  )+ \cR_V w  \rangle - \sum_{j=1}^{2N+2} c_j \langle i e_j \partial_{x_k}\eta_\mu, w    \rangle + \langle \partial_{x_k} \eta_\mu, \cR_V \eta_\mu\rangle ),
\end{equation}
\begin{align}
\partial_t \gamma = & \mu -\frac{1}{4}\|\bv\|^2 + \frac{1}{2}\partial_t \ba\cdot \bv - \lambda V_h(\ba,t;\omega) - \frac{1}{m'(\mu)} (\langle \partial_\mu\eta_\mu, N_\mu (w  ) + \cR_V(w  +\eta_\mu) \rangle \nonumber \\ & -\sum_{j=1}^{2N+2} c_j \langle i e_j \partial_\mu \eta_\mu, w  \rangle ),\label{eq:g}
\end{align}
\begin{equation}
\partial_t \mu =\frac{1}{m'(\mu)} \langle i\eta_\mu , N_\mu (w  ) -i \cR_V w   \rangle -\sum_{j=1}^{2N+2} c_j \langle e_j \eta_\mu ,w  \rangle .\label{eq:m}
\end{equation}

The claim of the lemma follows directly from Assumptions (A1) and (A3), which imply that 
\begin{equation*}
\sup_{\omega\in\overline{\Omega}}\|\cR_V E_k\|_{L^2} = O(\lambda h^2 ) , 
\end{equation*}
and 
\begin{equation*}
\sup_{\omega\in\overline{\Omega}}\|N_\mu (w  ) \|_{H^1} \le C ~\sup_{\omega\in\overline{\Omega}}\|w  \|^2_{H^1},
\end{equation*}
for $\sup_{\omega\in\overline{\Omega}} \|w  \|_{H^1}\le 1 $ and some constant $C$ that is independent of $h$ and $\lambda.$ 
 
\end{proof}

\subsection{Control of the fluctuation}

We use an approximate Lyapunov functional to obtain an explicit control on $\sup_{\omega\in\overline{\Omega}}\|w  \|_{H^1}$ and $\sup_{\omega\in\overline{\Omega}}|c|.$ This approach dates back to \cite{We1}, and has been used, with various generalizations, in \cite{FJGS1}-\cite{HZ2}. We define the Lyapunov functional 
\begin{equation}
\label{eq:LyapunovFunctional}
\cC_\mu (u,v):= \cE_\mu (u)- \cE_\mu (v) , \ \ u,v\in H^1(\bbR^N),
\end{equation}
where $\cE_\mu$ is given in (\ref{eq:EnergyFunctional}). 
We proceed by estimating upper and lower bounds for $\cC_\mu(u  ,\eta_\mu),$ $\omega\in\overline{\Omega},$ where $u  =\eta_\mu + w  $ is defined in (\ref{eq:CMSol}).

\begin{lemma}\label{lm:UpperBoundLF}
Suppose the hypotheses of Lemma \ref{lm:RepEqMotion} hold. Then there exists a constant $C$ independent of $h$ such that 
\begin{equation}
\label{eq:UpperBoundLF}
\sup_{\omega\in\overline{\Omega}}|\partial_t \cC_\mu (u  ,\eta_\mu)|\le C~ \sup_{\omega\in\overline{\Omega}}  \left (\lambda h^2 \|w  \|_{H^1} + (|c|+\lambda h + \|w\|_{H^1}^2) \|w  \|^2_{H^1}\right ),
\end{equation}
where $|c|$ appears in (\ref{eq:Alpha}), uniformly in $t\in (0,T)$ and $\lambda\in (h^{1-\epsilon},1].$

\end{lemma}

\begin{proof}
For $\omega\in\overline{\Omega},$ we have
\begin{align*}
\cE_\mu (u  ) &= \frac{1}{2} \int d\bx ~ |\nabla u  |^2 + \mu |u  |^2 -F(u  )\\
&= H_{\lambda}(T_{\ba\bv\gamma}^{-1}\psi  ) + \frac{1}{2}\mu \|T_{\ba\bv\gamma}^{-1}\psi  \|_{L^2}^2 - \frac{\lambda}{2} \int d\bx~ V_h |T_{\ba\bv\gamma}^{-1}\psi  |^2.
\end{align*}
By translational symmetry, we have, for $\omega\in\overline{\Omega},$
\begin{equation*}
\|u  \|_{L^2}^2 = \|\psi  \|_{L^2}^2 , \ \ \int d\bx ~  V_h |u  |^2 = \int  V_h^{-\ba} |\psi  |^2,
\end{equation*}
where $V_h^{\bx_0}(\bx):=  V_h(\bx+\bx_0), \; \bx_0\in\bbR^N.$ Moreover, for $\omega\in\overline{\Omega},$
\begin{equation*}
H_{\lambda}(T_{\ba\bv\gamma}^{-1}\psi  ) = H_{\lambda}(\psi  )+ \frac{1}{2}(\frac{1}{4}\|\bv\|^2 + \mu)\|\psi  \|_{L^2}^2 - \frac{1}{2}\bv\cdot \langle i\psi  , \nabla \psi   \rangle + \frac{\lambda}{2}\int d\bx ( V_h^{-\ba} -  V_h)|\psi  |^2 ,
\end{equation*}
and hence 
\begin{equation}
\label{eq:CMEnergy}
\cE_\mu (u  )= H_{\lambda}(\psi  ) + \frac{1}{2}(\frac{1}{4}\|\bv\|^2 + \mu)\|\psi  \|_{L^2}^2 - \frac{1}{2}\bv\cdot \langle i\psi   ,\nabla \psi   \rangle - \frac{\lambda}{2}\int d\bx~   V_h |\psi  |^2.
\end{equation}
We have the following relationships for the rate of change of field energy and momenta, whose proof we give in Appendix B. For $\omega\in\overline{\Omega},$
\begin{equation}
\label{eq:RatePotential}
\partial_t \frac{1}{2} \int d\bx  ~ V_h |\psi  |^2 = \langle \nabla  V_h i\psi  , \nabla \psi   \rangle + \frac{1}{2} \langle \psi  , \partial_t  V_h\psi   \rangle , 
\end{equation}
\begin{equation}
\label{eq:RateEnergy}
\partial_t H_{\lambda}(\psi  ) = \frac{\lambda}{2} \langle \psi  , \partial_t V_h \psi  \rangle
\end{equation}
and Ehrenfest's theorem
\begin{equation}
\label{eq:Ehrenfest}
\partial_t \langle i\psi  , \nabla\psi  \rangle = -\langle \psi   ,\lambda\nabla V_h \psi  \rangle. 
\end{equation}
One can use (\ref{eq:NLSE}) to formally derive (\ref{eq:RatePotential}),(\ref{eq:RateEnergy}) and (\ref{eq:Ehrenfest}). Rigorously, one can prove them using a regularization scheme and a limiting procedure, and we refer the reader to the Appendix B  for a proof of these statements. Furthermore, it follows from gauge invariance of (\ref{eq:NLSE}) that the charge is ${\mathbb P}$-a.s. conserved,
\begin{equation}
\label{eq:ConservationCharge}
\partial_t\|\psi  \|_{L^2} = 0, \ \ \omega\in \overline{\Omega},
\end{equation} 
(Proposition \ref{pr:LWP} in Appendix A).
Differentiating (\ref{eq:CMEnergy}) with respect to $t$ and using (\ref{eq:RatePotential}-\ref{eq:ConservationCharge}) gives, for $\omega\in\overline{\Omega},$
\begin{align}
\partial_t \cE_\mu (u  ) &= \partial_t H_{\lambda}(\psi  ) +\frac{1}{2}(\frac{\partial_t \bv\cdot \bv}{2} + \partial_t \mu)\|\psi  \|_{L^2}^2 - \frac{1}{2}\partial_t \bv \cdot \langle i\psi  , \nabla \psi  \rangle + \frac{1}{2} \bv\cdot \langle \psi   ,\lambda \nabla  V_h \psi  \rangle \nonumber \\ &- \langle \lambda \nabla  V_h i\psi  , \nabla \psi   \rangle - \frac{\lambda}{2} \langle \psi  , \partial_t  V_h\psi   \rangle  \nonumber \\
&= \frac{1}{2}(\frac{\partial_t \bv\cdot \bv}{2} + \partial_t \mu)\|\psi  \|_{L^2}^2- \frac{1}{2}\partial_t \bv \cdot \langle i\psi  , \nabla \psi  \rangle+ \frac{1}{2} \bv\cdot \langle \psi ,\lambda \nabla  V_h \psi  \rangle- \langle \lambda \nabla  V_h i\psi  , \nabla \psi   \rangle \nonumber \\
&= \frac{1}{2}\partial_t \mu \|u  \|_{L^2}^2 - \langle \frac{1}{2} \langle i(\partial_t \bv + 2 \lambda \nabla V_h^{\ba})u  ,\nabla u  \rangle, \label{eq:RateEnergyFunctional} 
\end{align}
where we have used $u  (\bx,t)=e^{i(\frac{1}{2}\bv\cdot \bx +\gamma)} \psi  (\bx+\ba,t)$ and translation invariance of the integral in the last line. 
Furthermore, it follows from (\ref{eq:EnergyFunctional}) that 
\begin{equation*}
\partial_t \cE_\mu(\eta_\mu ) = \frac{1}{2} \partial_t \mu \|\eta_\mu\|^2, \ \ \omega\in\overline{\Omega},
\end{equation*}
which, together with (\ref{eq:LyapunovFunctional}) and (\ref{eq:RateEnergyFunctional}), implies
\begin{equation}
\label{eq:RateLyapunov}
\partial_t \cC_\mu(u  ,\eta_\mu) = \frac{1}{2} \partial_t \mu (\|u  \|^2_{L^2}- \|\eta_\mu\|^2_{L^2}) - \langle i(\frac{1}{2}\partial_t \bv + \lambda  \nabla V_h^{\ba}) u   , \nabla u   \rangle , 
\end{equation}
for $\omega\in\overline{\Omega}.$ 
We now estimate the right-hand side of (\ref{eq:RateLyapunov}). Since $\langle i X, w  \rangle = 0$ for all $X\in \cT_\mu \cM_s$ and $\omega\in\overline{\Omega},$ we have that 
\begin{equation*}
\|u  \|^2_{L^2}- \|\eta_\mu\|^2_{L^2} = \|w  \|^2_{L^2},
\end{equation*}
and hence
\begin{equation}
\label{eq:FTEstimate}
\frac{1}{2} \partial_t \mu (\|u  \|^2_{L^2}- \|\eta_\mu\|^2_{L^2}) = O(|c| \|w  \|^2_{L^2}).
\end{equation}

To estimate the second term in the right-hand side of (\ref{eq:RateLyapunov}), we replace $u  =\eta_\mu+w  ,$ and use the fact that $\langle iE_g,w\rangle=\langle i \nabla \eta_\mu, w  \rangle = 0,$ and $\langle ig\eta_\mu, \nabla \eta_\mu \rangle = 0$ for all real $g\in L^\infty(\bbR^N)$ and $\omega\in\overline{\Omega}.$ We have 
\begin{equation*}
\langle i(\frac{1}{2}\partial_t \bv + \lambda\nabla  V_h^{\ba}) u   , \nabla u   \rangle = \langle i\lambda \nabla  V_h^{\ba} w  , \nabla \eta_\mu\rangle + \langle i \lambda \nabla  V_h^{\ba} \eta_\mu , \nabla w  \rangle + \langle i (\frac{1}{2}\partial_t \bv + \lambda \nabla V_h^{\ba})w  , \nabla w  \rangle,
\end{equation*}
for $\omega\in\overline{\Omega}.$
Adding and subtracting the quantity $$\lambda \nabla  V_h(\ba,t;\omega)\cdot \langle iw  ,\nabla \eta_\mu \rangle =  \lambda \nabla  V_h(\ba,t;\omega) \cdot \langle i\eta_\mu, \nabla w   \rangle = 0$$ gives
\begin{align*}
&\langle i(\frac{1}{2}\partial_t \bv + \lambda \nabla  V_h^\ba) u   , \nabla u   \rangle = 
(\frac{1}{2}\partial_t \bv + \lambda \nabla  V_h(\ba,t;\omega)) \langle iw  , \nabla w  \rangle + \langle (\lambda \nabla  V_h^\ba -\lambda \nabla  V_h(\ba,t;\omega))iw  ,\nabla w  \rangle  \\
&+ \langle (\lambda \nabla  V_h^\ba -\lambda \nabla  V_h(\ba,t;\omega))i\eta_\mu,\nabla w  \rangle + \langle (\lambda \nabla  V_h^\ba -\lambda \nabla  V_h(\ba,t;\omega))iw  ,\nabla \eta_\mu\rangle,
\end{align*}
for $\omega\in\overline{\Omega}.$
It follows from Lemma \ref{lm:RepEqMotion} that the first term of the above equation is of order $O(|c|\|w\|^2_{H^1}+\lambda h^2\|w\|_{H^1}^2+ \|w\|_{H^1}^4),$ while Assumption (B1) implies that the second term is of order $O(\lambda h \|w  \|_{H^1}^2).$ Assumptions (B1) and (A3) imply that the third and forth terms are of order $O(\lambda h^2 \|w  \|_{H^1}).$

\end{proof}

The object of the next lemma is to provide a lower bound for $\sup_{\omega\in\overline{\Omega}}|\cC_\mu (u  ,\eta_\mu)|.$ Let 
\begin{equation*}
X_\mu := \{w\in H^1(\bbR^N) : \ \ \langle w, i X\rangle =0, \forall X\in \cT_{\eta_\mu}\cM_s \}.
\end{equation*}
It follows from the coercivity property of $\cL_\mu$ that there exists a positive constant 
\begin{equation}
\label{eq:Coercivity}
\rho:= \inf_{w\in X_\mu}\langle w,\cL_\mu w\rangle >0,
\end{equation}
(see Appendix D in \cite{FJGS1} for a proof of this statement). 

\begin{lemma}\label{lm:LFLowerBound}
Suppose the hypotheses of Lemma \ref{lm:RepEqMotion} hold. Then there exists positive constants $\rho$ and $\overline{C}$ independent of $h$ and $\omega$ such that, for $\sup_{\omega\in\overline{\Omega}}\|w  \|_{H^1}\le 1$ and uniformly in $t\in [0,T],$
\begin{equation}
\label{eq:LFLowerBound}
\sup_{\omega\in\overline{\Omega}}|\cC_\mu (u  ,\eta_\mu)| \ge \sup_{\omega\in\overline{\Omega}}\{\frac{\rho}{2} \|w  \|_{H^1} - \overline{C} \|w  \|_{H^1}^3\},
\end{equation}
where $\rho$ is defined in (\ref{eq:Coercivity}).
\end{lemma}

\begin{proof}
For $\omega\in\overline{\Omega},$ we expand $\cE_\mu(u  )$ around $\eta_\mu,$ which is a critical point of $\cE_\mu.$
\begin{equation}
\cE_\mu (\eta_\mu + w  ) = \cE_\mu (\eta_\mu) + \frac{1}{2}\langle w  , \cL_\mu w  \rangle + R_\mu^{(3)}(w  ),
\end{equation}
where 
\begin{equation*}
R_\mu^{(3)}(w  )= F(\eta_\mu +w  ) - F(\eta_\mu)-\langle F'(\eta_\mu),w  \rangle -\frac{1}{2} \langle F''(\eta_\mu)w  ,w  \rangle.
\end{equation*}
It follows from Assumption (A1) that 
\begin{equation*}
\sup_{\omega\in\overline{\Omega}}|R_\mu^{(3)}(w  )| \le C \sup_{\omega\in\overline{\Omega}}\|w  \|_{H^1}^3,
\end{equation*}
for $\sup_{\omega\in\overline{\Omega}}\|w  \|_{H^1}\le 1,$ where $C>0$ is independent of $t\in \bbR.$ Moreover, the coercivity property (\ref{eq:Coercivity}) implies 
\begin{equation*}
\sup_{\omega\in\overline{\Omega}}\langle w  ,\cL_\mu w  \rangle \ge \rho \sup_{\omega\in\overline{\Omega}}\|w  \|_{H^1}^2,
\end{equation*}
and hence
\begin{equation*}
\sup_{\omega\in\overline{\Omega}}|\cC_\mu (u  ,\eta_\mu)| = \sup_{\omega\in\overline{\Omega}}|\cE_\mu(u  )-\cE_\mu (\eta_\mu)| \ge \sup_{\omega\in\overline{\Omega}}\{ \frac{1}{2}\rho \|w  \|_{H^1}^2 - \overline{C}\|w  \|_{H^1}^3\},
\end{equation*}
for $\sup_{\omega\in\overline{\Omega}}\|w  \|_{H^1}\le 1.$ 
\end{proof}

We now use the upper and lower bounds on the Lyapunov functional to obtain an upper bound on $\sup_{\omega\in\overline{\Omega}}\|w  \|_{H^1}.$ 

\begin{lemma}\label{lm:UpperBdFluctuation}
Suppose (A1)-(A5) and (B1) hold. Let $\psi $  satisfy (\ref{eq:NLSE}) and,  for $\omega\in\overline{\Omega},$ $u  ,\eta_\mu,w  $ be given as above.  For $h\ll 1,$ choose $T\in \bbR^+$ such that, for $\omega\in\overline{\Omega},$ $\psi  (t)\in U_{\delta}, \ \ t\in [0,T],$ where $U_\delta$ is defined in (\ref{eq:SODNbr}), Subsection \ref{sec:SOD}. Fix $\epsilon \in (0,1),$ and choose $t_0\in [0,T]$ such that $\sup_{\omega\in\overline{\Omega}}\|w  (t_0)\|_{H^1}^2 < h^{2-\epsilon}.$ Then, for $h$ small enough, there exist absolute constants $C_1>1$ and $C_2>0,$ which are independent of $h$ and $\epsilon,$ such that
\begin{align*}
\sup_{\omega\in\overline{\Omega}}\sup_{t\in [t_0,t_0+C_2/\lambda h ]}\|w  (t)\|_{H^1}^2 &\le C_1 (h^2 + \sup_{\omega\in\overline{\Omega}}\|w  (t_0)\|_{H^1}^2)) \\
\sup_{\omega\in\overline{\Omega}} \sup_{t\in [t_0,t_0+C_2/\lambda h ]} |c(t)| &\le C_1 (h^2 + \sup_{\omega\in\overline{\Omega}}\|w(t_0)\|_{H^1}^2) ,
\end{align*}
uniformly in $\lambda\in (h^{1-\epsilon},1].$
\end{lemma}

\begin{proof}
It follows from Lemma \ref{lm:UpperBoundLF} that, for $t\ge t_0$ 
\begin{align*}
&\sup_{\omega\in\overline{\Omega}}|\cC_\mu (\eta_\mu+w  (t),\eta_\mu)| \le \sup_{\omega\in\overline{\Omega}}\{|\cC_\mu (\eta_\mu+ w  (t_0),\eta_\mu)| \\ & + C(t-t_0) (|c| \|w  (t)\|^2_{H^1} +\lambda h^2 \|w  (t)\|_{H^1} + \lambda h\|w  (t)\|^2_{H^1} + \|w\|_{H^1}^4) \}.
\end{align*}
Expanding $\cE_\mu(\eta_\mu + w  (t_0))$ around $\eta_\mu$ and using Assumption (A1) gives the upper bound
\begin{equation*}
|\cC_\mu(\eta_\mu+w  (t_0), \eta_\mu)| \le C\|w  (t_0)\|_{H^1}^2 , \ \ \mathrm{for }\ \  \|w  (t_0)\|_{H^1}<1 \ \ \mathrm{and} \ \ \omega\in\overline{\Omega}, 
\end{equation*} 
where $C$ is a constant independent of $h,\epsilon$ and $\lambda.$  Therefore,
\begin{equation*}
\sup_{\omega\in\overline{\Omega}} |\cC_\mu (\eta_\mu+ w  (t),\eta_\mu)| \le C\sup_{\omega\in\overline{\Omega}} \|w  (t_0)\|_{H^1}^2 + C(t-t_0) \sup_{\omega\in\overline{\Omega}} (\lambda h^2 \|w  (t)\|_{H^1} + (|c|+\lambda h+ \|w\|_{H^1}^2) \|w  (t)\|_{H^1}^2),
\end{equation*}
for some constant $C$ independent of $h,$ $\epsilon$ and $\lambda.$ Together with Lemma \ref{lm:LFLowerBound}, it follows that
\begin{align*}
\frac{1}{2}\sup_{\omega\in\overline{\Omega}} \rho \|w  (t)\|_{H^1}^2 &\le C\sup_{\omega\in\overline{\Omega}}\|w  (t_0)\|_{H^1}^2  + C(t-t_0) \sup_{\omega\in\overline{\Omega}}(\lambda h^2 \|w  (t)\|_{H^1} + (|c|+\lambda h+ \|w\|_{H^1}^2) \|w  (t)\|_{H^1}^2) \\ &+ C\sup_{\omega\in\overline{\Omega}}\|w  (t)\|_{H^1}^3 
\end{align*}
where $\rho$ appears in (\ref{eq:Coercivity}). Equivalently, there exists a positive constant $C$ independent of $h, \epsilon$ and $\lambda,$ such that 
\begin{align*}
C\sup_{\omega\in\overline{\Omega}} \|w  (t)\|_{H^1}^2 &\le \sup_{\omega\in\overline{\Omega}}\|w  (t_0)\|_{H^1}^2 + (t-t_0)\sup_{\omega\in\overline{\Omega}}(\lambda h^2 \|w  (t)\|_{H^1} + (|c|+\lambda h+\|w\|_{H^1}^2) \|w  (t)\|_{H^1}^2 )\\ & + \sup_{\omega\in\overline{\Omega}}\|w  (t)\|_{H^1}^3.
\end{align*}
For $$t- t_0 \le \frac{C}{2(\lambda h+|c| + \sup_{\omega\in\overline{\Omega}}\|w\|_{H^1}^2)}=:\tau,$$
\begin{equation*}
C\sup_{\omega\in\overline{\Omega}}\|w  (t)\|_{H^1}^2 \le \sup_{\omega\in\overline{\Omega}}\{\|w  (t_0)\|^2_{H^1} + \frac{C}{2}h \|w  (t)\|_{H^1} + \frac{C}{2}\|w  (t)\|_{H^1}^2 + \|w  (t)\|_{H^1}^3\}.
\end{equation*}
Using the fact that 
\begin{equation*}
h \|w  (t)\|_{H^1} \le \frac{1}{2}h^2 + \frac{1}{2}\|w  (t)\|_{H^1}^2,
\end{equation*}
we have
\begin{equation*}
\sup_{\omega\in\overline{\Omega}}\|w  (t_0)\|_{H^1}^2 + \frac{C}{4} h^2 - \frac{C}{4}\sup_{\omega\in\overline{\Omega}} \|w  (t)\|_{H^1}^2 +\sup_{\omega\in\overline{\Omega}} \|w  (t)\|_{H^1}^3 \ge 0.
\end{equation*}
Let $y_0 := \sup_{\omega\in\overline{\Omega}}\|w  (t_0)\|_{H^1},$ $y:= \sup_{\omega\in\overline{\Omega}}\sup_{t\in [t_0, t_0 +\tau ]}\|w  (t)\|_{H^1},$ and 
$f(y) = y^3 - \frac{C}{4}y^2 + y_0^2 + \frac{C}{4}h^2.$ For $h\ll 1,$ the function intersects the x-axis in a point $y_*$ such that $y_*^2< c_1(h^2+ y_0^2),$ where $c_1$ is a positive constant independent of $h$ and $\epsilon.$ It follows, if $y_0<y_*,$ that $y<y_*$ for $t\in [t_0, t_0+\tau].$ Substituting back in (\ref{eq:a}-\ref{eq:m}) and using (\ref{eq:Coefficients}) and (\ref{eq:Alpha}) gives 
\begin{equation*}
\sup_{\omega\in\overline{\Omega}}|c| \le C_2' (h^2 + y_0^2),
\end{equation*}   
for some positive constant $C_2'$ that is independent of $h, \epsilon$ and $\lambda.$ It follows that for $h$ small enough, there exists positive constants $C_1$ and $C_2$ which are independent of $h$ and $\epsilon,$ such that
\begin{align*}
&\sup_{\omega\in\overline{\Omega}}\; \sup_{t\in [t_0,t_0 + \frac{C_2}{\lambda h}] } \|w  (t)\|_{H^1}^2 \le C_1 (h^2 +\sup_{\omega\in\overline{\Omega}} \|w  (t_0)\|^2_{H^1}) \\
&\sup_{\omega\in\overline{\Omega}}\sup_{t\in [t_0,t_0 + \frac{C_2}{\lambda h}]} |c(t)| \le C_1 (h^2 +\sup_{\omega\in\overline{\Omega}} \|w  (t_0)\|^2_{H^1}),
\end{align*} 
uniformly in $\omega\in\overline{\Omega}$ and $\lambda\in (h^{1-\epsilon},1].$
\end{proof}

\subsection{Proof of Proposition \ref{pr:Main}}\label{sec:ProofPrMain}

We are now in a position to prove Proposition \ref{pr:Main} by iterating Lemma \ref{lm:UpperBdFluctuation} on time intervals of order $O(\frac{1}{\lambda h})$ and using the result of Lemma \ref{lm:RepEqMotion}.

\begin{proof}[Proof of Proposition \ref{pr:Main}]
Fix $\epsilon\in (0,1).$ For $\omega\in\overline{\Omega}$ and $h$ small enough, let $T^*$ be the maximal time for which the skew-orthogonal decomposition is possible. Consider the interval 
\begin{equation*}
[0,T] = [t_0,t_1]\cup [t_1,t_2] \cup \cdots \cup [t_{n-1},t_n] \subset [0,T^*],
\end{equation*}
such that 
$$0=t_0< t_1 < \cdots < t_n=T, \ \ (t_{i+1}-t_i) \le \frac{C_2}{\lambda h}, \ \ i=0,\cdots, n-1,$$ 
where $C_2$ appears is Lemma \ref{lm:UpperBdFluctuation}. We will choose $n\in {\mathbb N}$ depending on $h$ and $\epsilon$ later. Let 
\begin{align*}
|c|_i &:= \sup_{\omega\in\overline{\Omega}}\; \sup_{t\in [t_i,t_{i+1}]} |c(t)| ,\\
y_i &:= \sup_{\omega\in\overline{\Omega}}\; \sup_{t\in [t_i,t_{i+1}]} \|w  (t)\|_{H^1}, \ \ i=0,\cdots , n-1.
\end{align*}
Note that $y_0 \le h$ and $|c|_0 \le C h,$ for some constant $C$ independent of $h,\lambda$ and $\epsilon.$ Iterating the application of Lemma \ref{lm:UpperBdFluctuation} we have 
\begin{align*}
y_n^2 &\le (\sum_{j=1}^n C_1^j) h^2 \le C_1^{n+1} h^2  \\
|c|_n &\le C C_1^{n+1}h^2.
\end{align*}
We now choose $n$ such that $C_1^{n+1}h^2 \le h^{2-\epsilon}.$ This implies 
\begin{equation*}
n+1 \le -\epsilon \frac{\log h}{\log C_1}. 
\end{equation*}
Therefore, for $t\in [0, \epsilon \frac{C_2}{\log C_1} \frac{|\log h|}{\lambda h}],$
\begin{align*}
&\sup_{\omega\in\overline{\Omega}}\|w  (t)\|_{H^1}^2 \le h^{2-\epsilon} \\
&\sup_{\omega\in\overline{\Omega}}|c(t)| \le C h^{2-\epsilon}.
\end{align*}
The effective equations for the parameters on the soliton manifold follow from the above estimates and Lemma \ref{lm:RepEqMotion}. Furthermore, it follows from (\ref{eq:InitialCond}) and the skew-orthogonal property that $$\|\ba(0)-\ba_0\|,\|\bv(0)-\bv_0\|, |\gamma(0)-\gamma_0|, |\mu(0)-\mu_0| = O(h).$$
 
\end{proof}

Equipped with the above results, we are in position to prove Theorem \ref{th:Main3}.  

\subsection{Proof of Theorem \ref{th:Main3}}

In Proposition \ref{pr:Main}, the time-dependent parameters $\sigma(t)=(\ba(t),\bv(t),\gamma(t),\mu(t))$ of the soliton solution satisfying equations (\ref{eq:EffectiveA})-(\ref{eq:EffectiveMu}) are defined for $\omega\in\overline{\Omega}$ (a dense set of $\Omega$). We extend them to be random variables of all the realization space $\Omega.$ 

\begin{proof}[Proof of Theorem \ref{th:Main3}]
We extend $(\ba(t),\bv(t),\gamma(t),\mu(t))_{0\le t\le \overline{C} \epsilon |\log h|/\lambda h}$ appearing in Proposition \ref{pr:Main} to random variables on $\Omega$ by requiring that they satisfy 
\begin{align*}
\partial_t a_k =& v_k + \xi(\omega)\frac{1}{m(\mu)}(\langle i x_k\eta_\mu,  N_\mu (w)-i\cR_Vw \rangle + \sum_{j=1}^{2N+2} c_j \langle e_j x_k \eta_\mu, w   \rangle) , \\ 
\partial_t v_k =& -2 \partial_{x_k} \lambda V_h(\ba,t;\omega) + \xi(\omega)[\frac{2}{m(\mu)} (\langle \partial_{x_k}\eta_\mu, N_\mu (w  )+ \cR_V w  \rangle \\ & -\sum_{j=1}^{2N+2} c_j \langle i e_j \partial_{x_k}\eta_\mu, w    \rangle + \langle \partial_{x_k} \eta_\mu, \cR_V \eta_\mu\rangle )],\\
\partial_t \gamma = & \mu -\frac{1}{4}\|\bv\|^2 + \frac{1}{2}\partial_t \ba\cdot \bv - \lambda V_h(\ba,t;\omega) - \xi(\omega) \frac{1}{m'(\mu)} (\langle \partial_\mu\eta_\mu, N_\mu (w  ) + \cR_V(w  +\eta_\mu) \rangle \nonumber \\ & -\xi(\omega)\sum_{j=1}^{2N+2} c_j \langle i e_j \partial_\mu \eta_\mu, w  \rangle ) \\
\partial_t \mu =&\xi(\omega) \frac{1}{m'(\mu)} \langle i\eta_\mu , N_\mu (w  ) -i \cR_V w   \rangle -\sum_{j=1}^{2N+2} c_j \langle e_j \eta_\mu ,w  \rangle 
\end{align*}
where 
\begin{equation*}
\xi(\omega) = 
\begin{cases}
1, \ \ \omega\in\overline{\Omega}\\
0, \ \ \omega\in \Omega\backslash \overline{\Omega}
\end{cases},
\end{equation*}
with the same intial conditions appearing in Proposition \ref{pr:Main}, which are determined by the skew-orthogonal decomposition property.
Since $V$ is also $\omega$-measurable, $\sigma(t)= (\ba(t),\bv(t),\gamma(t),\mu(t))$  and $\eta_{\sigma(t)}$ are $\omega$-measurable for $0\le t\le \overline{C} \epsilon |\log h|/\lambda h.$
Note that we have from Proposition \ref{pr:Main} that, for $0\le t\le \overline{C} \epsilon |\log h|/\lambda h,$
\begin{equation*}
{\mathbb E}\| \psi(t) - \eta_{\sigma(t)}\|_{H^1}^2 \le C h^{2-\epsilon},
\end{equation*} 
which completes the proof of the theorem.
\end{proof}

%%%%%%%%%%%%%%%%%%%%%%%%%%%%%%%%%%%%%%%%%%%%%%%%
%A LIMIT THEOREM
%%%%%%%%%%%%%%%%%%%%%%%%%%%%%%%%%%%%%%%%%%%%%%%%

\section{A limit theorem in the weak-coupling/space-adiabatic regime}\label{sec:ProofMain4}

In this subsection, we prove Theorems \ref{th:Main4} and \ref{th:Main5}, which are limit theorem for the dynamics of the center of mass of the soliton moving in $N\ge 2,$ respectively $N\ge3,$ under the influence of a homogeneous, time-independent and strongly mixing potential satisfying, in addition to assumption (B1), assumption (B2), Section  \ref{sec:MainResults}. The main ingredient of our analysis is proving (strong) convergence of the center of mass dynamics to the auxiliary dynamics of a classical particle in the weak-coupling/space-adiabatic limit, Lemmatta \ref{lm:StrongConv} and \ref{lm:ConvergenceLiouvilleSol} below. This allows us to use the results of \cite{KP}, \cite{KR1} and \cite{KR2} for the limiting dynamics of a classical point particle in a random potential. For the convenience of the reader, we summarize the main conceptual ideas in the analysis of \cite{KP}, \cite{KR1} and \cite{KR2} in the following subsection.

\subsection{Proof of Theorems \ref{th:Main4} and \ref{th:Main5}}

We know from Theorem \ref{th:Main3}, that up to times of order $O(|\log h|/\lambda h),$ the soliton center of mass behaves like a point particle over spatial and temporal scales of order $O(h^{-1}).$ This motivates introducing the scaling 
\begin{align*}
& \overline{\ba}:= h\ba\\
&\overline{\bv} := \bv\\
&\overline{t} := ht.
\end{align*}
It follows from Theorem \ref{th:Main3} that the rescaled dynamics of the center of mass of the soliton, for $\overline{t}\in [0,\overline{C}\epsilon |\log h|/\lambda),$ is given by 
\begin{align}
&\partial_{\overline{t}}\overline{\ba} = \overline{\bv} + O(h^{2-\epsilon})\label{eq:Res1a}\\
&\partial_{\overline{t}}\overline{\bv} = -2\lambda \nabla \overline{V}(\overline{\ba};\omega) + O(h^{1-\epsilon}), \label{eq:Res1v}
\end{align}
with initial condition satisfying
\begin{align}
&\|\overline{\ba}(0) - h\ba_0\| = O(h^2)\\
&\|\overline{\bv}(0)-\bv_0\|=O(h).\label{eq:Res1InitCond}
\end{align}

We introduce  the auxiliary dynamics corresponding to a classical particle in the random potential,
\begin{align}
&\partial_{\overline{t}} \tilde{\ba} = \tilde{\bv} \label{eq:Res2a}\\
&\partial_{\overline{t}} \tilde{\bv} = -2\lambda \nabla \overline{V}(\tilde{\ba};\omega),\label{eq:Res2v}
\end{align}
with initial condition 
\begin{equation}
\label{eq:Res2InitCond}
\tilde{\ba}(0)= 0, \ \ \tilde{\bv}(0)=\bv_0.
\end{equation}
This dynamics is ``close'' to the effective dynamics of the solitary wave in the limit $h\rightarrow 0.$ Moreover, it has been studied extensively as a model of stochastic acceleration of classical particles in a random potential, \cite{KP}-\cite{KR2}.  

We note that momentum diffusion occurs at scales of order $O(\lambda^{-2}),$ while spatial diffusion occurs at scales of order $O(\lambda^{-2-\beta}),$ for $\beta>0.$ This motivates comparing the limiting behavior of the true dynamics of the center of mass of the soliton to the  auxiliary one corresponding to a classical point particle in the external potential. We have the following lemma about the convergence of the effective dynamics of the center of the soliton to the auxiliary dynamics over scales $O(\lambda^{-2}).$ 

\begin{lemma}\label{lm:StrongConv}
Assume (A1)-(A5), (B1) and (B2), and suppose that there exists $\tilde{\alpha}>0$ such that  $\lambda\rightarrow 0$ as $h\rightarrow 0$ with $|\log h|\lambda^{3/2+\tilde{\alpha}}\rightarrow \infty.$ Then, for any finite $T >0,$ the stochastic process  $$(\lambda^2\overline{\ba}(\overline{t}/\lambda^{2}), \overline{\bv}(\overline{t}/\lambda^2))_{\overline{t}\in [0,T ]}$$ converges ${\mathbb P}$-a.s. (strongly) to the stochastic process $$(\lambda^2 \tilde{\ba}(\overline{t}/\lambda^{2}), \tilde{\bv}(\overline{t}/\lambda^{2}))_{\overline{t}\in [0,T ]},$$ as $\lambda,h\rightarrow 0.$
\end{lemma}

\begin{proof}

We fix finite $T >0,$ and choose $\lambda$ and $h$ small enough such that $T/\lambda^2<\overline{C}\epsilon |\log h|/\lambda.$ We define  
$$f_{\lambda,h}(\overline{t}):= \lambda^{2} \|\overline{\ba}(\overline{t}/\lambda^{2})-\tilde{\ba}(\overline{t}/\lambda^{2})\|, $$ 
and 
$$g_{\lambda,h}(\overline{t})=\|\ov(\overline{t}/\lambda^2)-\tv(\overline{t}/\lambda^2)\|, $$
for $\overline{t}\in [0,T ].$ We have from (\ref{eq:Res1a}) - (\ref{eq:Res2InitCond}) that
\begin{equation}
\label{eq:DiffInq1}
|\partial_{\overline{t}} f_{\lambda,h}(\overline{t})|  \le C ( g_{\lambda ,h}(\overline{t}) + h^{2-\epsilon})
\end{equation}
and 
\begin{align}
|\partial_{\overline{t}} g_{\lambda ,h}| & \le C( \frac{1}{\lambda}\|\oa(\overline{t}/\lambda^2) -\ta(\overline{t}/\lambda^2)\| \sup_{\omega\in\overline{\Omega}}\|\nabla^2 \overline{V}\|_{L^\infty} + h^{1-\epsilon}\lambda^{-2})\nonumber \\
&\le C (\frac{1}{\lambda^3}f_{\lambda,h} + h^{1-\epsilon}\lambda^{-2}), \label{eq:DiffInq2}
\end{align}
where $C$ is a positive constant that is independent of $h,\epsilon$ and $\lambda.$
We now use Gronwall's lemma and Duhamel  formula to obtain bounds on $f_{\lambda ,h}$ and $g_{\lambda ,h}.$

Introduce the auxiliary $C^1(\bbR)$ functions $\tilde{f}_{\lambda, h}$ and $\tilde{g}_{\lambda,h}$ that satisfy 
\begin{align}
|\partial_{\overline{t}} \tilde{f}_{\lambda,h}(\overline{t})|& \le C \tilde{g}_{\lambda, h}(\overline{t})\\
|\partial_{\overline{t}} \tilde{g}_{\lambda,h}(\overline{t})| & \le C \frac{1}{\lambda^3} \tilde{f}_{\lambda, h}(\overline{t}),\label{eq:AuxDiffInq}
\end{align}
with initial conditions $\tilde{f}_{\lambda, h}(0) = f_{\lambda, h}(0)$ and $\tilde{g}_{\lambda, h}(0)= g_{\lambda, h}(0).$
We also introduce the rescaled $C^1(\bbR)$ functions $\overline{f}$ and $\overline{g}$ defined by 
\begin{align}
\overline{f}_{\lambda,h}(\overline{t})&:= \lambda^{-3/2} \tilde{f}_{\lambda,h}(\lambda^{3/2}\overline{t})\\
\overline{g}_{\lambda,h}(\overline{t})&:= \tilde{g}_{\lambda,h}(\lambda^{3/2}\overline{t}).\label{eq:AuxiliaryScaling}
\end{align}
They satisfy the differential inequalities
\begin{align*}
|\partial_{\overline{t}} \overline{f}_{\lambda,h}(\overline{t})|& \le C \overline{g}_{\lambda, h}(\overline{t})\\
|\partial_{\overline{t}} \overline{g}_{\lambda,h}(\overline{t})| & \le C \overline{f}_{\lambda, h}(\overline{t}).
\end{align*}
We define $$\overline{l}_{\lambda, h}:= \overline{f}_{\lambda ,h}^2 + \overline{g}_{\lambda, h}^2.$$ It follows from the above inequalities that 
\begin{equation*}
|\partial_{\overline{t}} \overline{l}_{\lambda, h}| \le 2C \overline{l}_{\lambda, h}.
\end{equation*}
By Gronwall's lemma, we have 
\begin{equation*}
\sup_{\overline{t}\in [0,T]}\overline{l}_{\lambda,h}(\overline{t}) \le e^{2CT}\overline{l}_{\lambda, h}(0), 
\end{equation*}
which implies 
\begin{equation*}
\sup_{\overline{t}\in [0,T]}(|\overline{f}_{\lambda,h}(\overline{t})|, |\overline{g}_{\lambda,h}(\overline{t})|) \le  e^{C T} (|\overline{f}_{\lambda,h}(0)|+|\overline{g}_{\lambda,h}(0)|).
\end{equation*}
Using (\ref{eq:AuxiliaryScaling}), we have that 
\begin{equation}
\label{eq:AuxiliaryBd}
\sup_{\overline{t}\in [0,T]}(|\tilde{f}_{\lambda,h}(\overline{t})|, |\tilde{g}_{\lambda,h}(\overline{t})|) \le  C'e^{C T/\lambda^{3/2}} h ,
\end{equation}
for some constant $C'$ that is independent of $h,\epsilon$ and $\lambda.$ Now, it follows from (\ref{eq:DiffInq1})-(\ref{eq:AuxDiffInq}), (\ref{eq:AuxiliaryBd}) and the Duhamel formula that 
\begin{equation}
\label{eq:DifferenceBd}
\sup_{\overline{t}\in [0,T]}(|f_{\lambda,h}(\overline{t})|,| g_{\lambda,h}(\overline{t})| )\le  C'e^{CT/\lambda^{3/2}} h^{1-\epsilon}\lambda^{-2} .
\end{equation}
Since $|\log h|\lambda^{3/2+\tilde{\alpha}}\rightarrow \infty$ as $\lambda, h\rightarrow 0$ for some $\tilde{\alpha}>0,$ we have that  
\begin{equation}
\lim_{\lambda, h\rightarrow 0}e^{CT/\lambda^{3/2}} h^{1-\epsilon}\lambda^{-2}  =0,
\end{equation}
which, together with (\ref{eq:DifferenceBd}) imply the claim of the lemma. \end{proof}

This lemma allows us to apply the results of \cite{KP} and \cite{KR1} on momentum diffusion for weakly random Hamiltonian systems. 

\begin{proof}[Proof of Theorem \ref{th:Main4}]

We know from \cite{KP} that in dimensions $N\ge 3,$ and from \cite{KR1} in $N=2,$ that the auxiliary stochastic process 
$$(\lambda^2\oa(\overline{t}/\lambda^2) , \ov(\overline{t}/\lambda^2))_{\overline{t}\ge 0}$$ converges in law, as 
$\lambda\rightarrow 0,$ to the stochastic process $(\int_0^{\overline{t}} \underline{v}(s)ds, \underline{v}(\overline{t})),$ where $\underline{v}$ satisfy (\ref{eq:Diff1}). The claim follows now from Lemma \ref{lm:StrongConv} and the fact that ${\mathbb P}$-a.s. convergence implies convergence in law.
\end{proof}

We now specialize to the case $N\ge 3.$ We introduce the auxiliary Liouville equation corresponding to a classical particle moving in the random potential. Let $\tilde{\phi}^\lambda(\bx,t,\bk)$ satisfy the Liouville equation
\begin{equation}
\label{eq:Liouville2}
\partial_t \tilde{\phi}^\lambda = \partial_{\overline{t}} \ta|_{\ta=\bx,\tv=\bk }\cdot \nabla_{\bx} \tilde{\phi}^\lambda + \partial_{\overline{t}} {\tv}|_{\ta=\bx,\tv=\bk}\cdot \nabla_{\bk} \tilde{\phi}^\lambda = \bk\cdot \nabla_{\bx} \tilde{\phi}^\lambda - 2\lambda \nabla_\bx \overline{V}\cdot \nabla_\bk \tilde{\phi}^\lambda .
\end{equation}
with initial condition $\tilde{\phi}^\lambda(\bx,0,\bk)= \phi_0(\lambda^{2+\beta}\bx,\bk), \beta >0,$ where $\phi_0$ is compactly supported, twice differentiable in $\bx\in\bbR^N,$ and four times differentiable in $\bk\in\bbR^N,$ such that its support is contained in the shell
\begin{equation*}
{\mathcal A}(M):= \{(\bx,\bk)\in \bbR^{2N}, \ \ 1/M < \|\bk\|< M\},
\end{equation*}
for some $M>1.$
We now show that the solutions of (\ref{eq:Liouville1}) and (\ref{eq:Liouville2}) are almost surely in $ C^1(\bbR; C^1(\bbR^{2N})\cap W^{1,\infty}(\bbR^{2N})).$ 

\begin{lemma}\label{lm:SolLiouville}
Let $\phi^\lambda$ and $\tilde{\phi}^\lambda$ be solutions of the Liouville equations (\ref{eq:Liouville1}) and (\ref{eq:Liouville2}), respectively, with the same initial condition $\phi_0$ as described above. Then $$\phi^\lambda \ \ \mathrm{and} \ \ \tilde{\phi}^\lambda \ \ \in C^1(\bbR; C^1(\bbR^{2N})\cap W^{1,\infty}(\bbR^{2N})) \ \ {\mathbb P}-\mathrm{a.s.},
$$
for all $t\ge 0.$
\end{lemma}
\begin{proof}
Consider first the Liouville equation corresponding to the auxiliary dynamics, (\ref{eq:Liouville2}). The analysis for (\ref{eq:Liouville1}) is similar. Introduce the variables ${\mathbf X}:= (\bx,\bk)\in \bbR^{2N},$ and let $A({\mathbf X}) = (\bk, -2\lambda \nabla_{\bx}\overline{V}(\bx))|_{{\mathbf X}:= (\bx,\bk)},$ $\Phi ({\mathbf X},t):= \tilde{\phi}^\lambda (\bx,t,\bk).$  Note that $$A({\mathbf X})= A_{l}({\mathbf X}) + A_{n}({\mathbf X}),$$ where $A_{l}$ is linear in ${\mathbf X}$ and $A_{n}$ is nonlinear in ${\mathbf X},$ with $A_n\in C^1(\bbR^{2N})\cap W^{1,\infty}(\bbR^{2N})$ as a vector-valued function, since $\overline{V}\in W^{2,\infty}\cap C^2.$
Now, Eq. (\ref{eq:Liouville2}) can be written as the Hamilton-Jacobi equation 
\begin{equation*}
\partial_t \Phi(\mathbf X,t) = A({\mathbf X}) \cdot \nabla_{{\mathbf X}} \Phi({\mathbf X},t),
\end{equation*}
with initial condition $\Phi_{0}\in C^2(\bbR^{2N}).$
We use the method of characteristics to solve the above Hamilton-Jacobi equation; see for example \cite{L1} for a discussion about Hamilton-Jacobi equations. Consider the mapping 
$$\chi_{t,t_0} ({\mathbf X}):= {\mathbf X} - A({\mathbf X}) (t-t_0) \ \ \in \bbR^{2N}.$$ We want to show that ${\mathbb P}$-a.s. $\chi$ is a diffeomorphism of class $C^1,$  i.e., for $\omega\in\overline{\Omega},$ $\chi$ is bijective and with a $C^1$ inverse.
For $\omega\in\overline{\Omega}$ and $0<t-t_0< \frac{1}{2\sup_{\omega\in\overline{\Omega}}\|A_{n}\|_{L^\infty}},$ the mapping $\chi$ is invertible and differentiable, since the linear part $A_l$ is clearly invertible and differentiable, while the nonlinear part $A_n$ is a differentiable perturbation for short times.

We claim that for any $t-t_0>0,$ the mapping $\chi$ is a ${\mathbb P}$-a.s. diffeomorphism of class $C^1.$ Let $\tau>0$ be the maximal  time such that $\chi_{\tau+t_0,t_0}$ is ${\mathbb P}$-a.s. invertible with $C^1$ inverse, and suppose that $\tau<\infty.$ Let $\delta t:= \frac{1}{4\sup_{\omega\in\overline{\Omega}} \|A_{n}\|_{L^\infty}}.$ It follows from the definition of $\chi$ that  
$$\chi_{\tau +\frac{\delta t}{4}, t_0} = \chi_{\tau - \frac{\delta t}{4},t_0} + A_l \frac{\delta t}{2} + A_n\frac{\delta t}{2}.$$
Since $\tau- \frac{\delta t}{4}<\tau,$ $\chi_{\tau - \frac{\delta t}{4},t_0}$ is a $C^1$ diffeomorphism. Furthermore, the linear part is differentiable and invertible. Again, $A_n\frac{\delta t}{2}$ is a small differentiable perturbation, and hence $\chi_{\tau +\frac{\delta t}{4}, t_0}$ is a $C^1$ diffeomorphism. However, this contradicts the definition of $\tau,$ and hence $\tau=\infty.$

Using the definition of $\chi,$ it is straightforward to verify that the solution $\Phi$ of the Hamilton-Jacobi equation satisfies 
$$\Phi(\chi_{t,0}({\mathbf X}),t) = \Phi_0({\mathbf X}).$$ Since $\chi$ has a differentiable inverse ${\mathbb P}$-a.s., 
$$\Phi({\mathbf X},t) = \Phi_0(\chi^{-1}_{t,0}({\mathbf X})).$$ Furthermore, $\Phi_0\in C^2(\bbR^{2N}),$ implies that $\Phi\in  C^1(\bbR; C^1(\bbR^{2N})\cap W^{1,\infty}(\bbR^{2N})) \ \ {\mathbb P}$-almost surely. 

Using the same argument, one can show that the solution of (\ref{eq:Liouville1}) with initial condition $\phi_0$ is almost surely in $ C^1(\bbR; C^1(\bbR^{2N})\cap W^{1,\infty}(\bbR^{2N}))$ by studying the solution of the corresponding Hamilton-Jacobi equation. 
\end{proof}

We also have the following lemma on the convergence of the (rescaled) solutions of (\ref{eq:Liouville1}) and (\ref{eq:Liouville2}) in the weak-coupling/space-adiabatic limit.

\begin{lemma}\label{lm:ConvergenceLiouvilleSol}
Let $\phi^\lambda$ and $\tilde{\phi}^\lambda$ be solutions of the Liouville equations (\ref{eq:Liouville1}) and (\ref{eq:Liouville2}), respectively, with the same initial condition $\phi_0$ as described above. Suppose that there exists $\tilde{\alpha}>0$ such that $\lambda\rightarrow 0$ as $h\rightarrow 0$ with $|\log h|\lambda^{1+\tilde{\alpha}}\rightarrow \infty.$ Then, for any fixed $T>0$ and any $\beta\in (0,\tilde{\alpha}/2),$ we have that 
\begin{equation}
\label{eq:ClassicalLimit}
\lim_{\lambda, h\rightarrow  0}\sup_{(t,\bx,\bk)\in [0,T ]\times K} [\tilde{\phi}^\lambda(\bx/\lambda^{2+\beta},t/\lambda^{2+2\beta},\bk) - \phi^\lambda(\bx/\lambda^{2+\beta},t/\lambda^{2+2\beta},\bk)] =0 \ \ {\mathbb P}-\mathrm{a.s.}.
\end{equation}
\end{lemma} 
\begin{proof}
For fixed $0<T<\infty,$ let $\lambda, h$  be small enough such that $T\lambda^{-2-2\beta}<\overline{C}\epsilon |\log h|/\lambda.$ This is possible since  $|\log h|\lambda^{1+\tilde{\alpha}}\rightarrow \infty$ as $\lambda,h\rightarrow 0,$ and $\beta<\tilde{\alpha}/2.$
Recall that $\phi^\lambda$ and $\tilde{\phi}^\lambda$ satisfy the same initial conditions. Let 
$$\overline{\phi}^\lambda:= \phi^\lambda -\tilde{\phi}^\lambda.$$
It follows from (\ref{eq:Liouville1}), (\ref{eq:Liouville2}) and Lemma \ref{lm:SolLiouville} that,
\begin{equation*}
\partial_{t}\overline{\phi}^\lambda = \bk\cdot \nabla_\bx \overline{\phi}^\lambda -2\lambda \nabla_\bx \overline{V}(\bx) \cdot \nabla_\bk \overline{\phi}^\lambda + O(h^{1-\epsilon}\|\phi^\lambda\|_{W^{1,\infty}(\bbR^{2N})}),
\end{equation*} 
for $t\in [0,T],$ with initial condition $\overline{\phi}^\lambda(0)=0.$
Therefore, for $t\in [0,T ]$ and all $\omega\in\overline{\Omega},$
\begin{equation*}
|\overline{\phi}^\lambda (\bx/\lambda^{2+\beta},t/\lambda^{2+2\beta},\bk)| = O( T \lambda^{-4-3\beta} h^{1-\epsilon}).
\end{equation*}
Now, for $0<\beta<\tilde{\alpha}/2,$
\begin{equation*}
\lim_{\lambda, h\rightarrow 0} T \lambda^{-4-3\beta} h^{1-\epsilon} =  0,
\end{equation*}
which implies the claim of the lemma.
\end{proof}

We now use Lemma \ref{lm:ConvergenceLiouvilleSol} and the results of \cite{KR2} to prove Theorem \ref{th:Main5}.

\begin{proof}[Proof of Theorem \ref{th:Main5}]

It is shown in \cite{KR2} that momentum diffusion of a classical particle in a random and strongly mixing potential converges to a spatial Brownian motion over longer time scales (see Subsection \ref{sec:RevisionStochastic}). As a cosequence, the auxiliary stochastic process given in (\ref{eq:Res2a})-(\ref{eq:Res2InitCond}) converges in the weak-coupling limit to a spatial Brownian motion:
For $\lambda \ll 1,$ there exist $\tilde{\beta} \in (0,\tilde{\alpha}/2)$ and some constant $C$ independent of $h$ and $\lambda,$ such that, for all $0< \beta<\tilde{\beta},$ fixed $0<t_0< T $ and all compact sets $K\subset{\mathcal A}(M),$ we have 
\begin{equation}
\label{eq:ClassicalLimit2}
\sup_{(t,\bx,\bk)\in [t_0,T ]\times K} |{\mathbb E}[\tilde{\phi}^\lambda(\bx/\lambda^{2+\beta},t/\lambda^{2+2\beta},\bk)]- u(\bx,t,\bk) | \le C T  \lambda^{\tilde{\beta}-\beta},
\end{equation}
where $\tilde{\phi}^\lambda$ satisfied (\ref{eq:Liouville2}) and $u$ satisfies (\ref{eq:SpaceDiffusion}).
We also have from Lemma \ref{lm:ConvergenceLiouvilleSol} that 
\begin{equation}
\label{eq:ClassicalLimit3}
\lim_{\lambda, h\rightarrow  0}\sup_{(t,\bx,\bk)\in [t_0,T ]\times K} {\mathbb E}[\tilde{\phi}^\lambda(\bx/\lambda^{2+\beta},t/\lambda^{2+2\beta},\bk) - \phi^\lambda(\bx/\lambda^{2+\beta},t/\lambda^{2+2\beta},\bk) ]=0.
\end{equation}
It follows (\ref{eq:ClassicalLimit2}) and (\ref{eq:ClassicalLimit3}) that 
\begin{equation*}
\lim_{\lambda, h\rightarrow  0}\sup_{(t,\bx,\bk)\in [t_0,T ]\times K} |{\mathbb E}[\phi^\lambda(\bx/\lambda^{2+\beta},t/\lambda^{2+2\beta},\bk)]- u(\bx,t,\bk) |=0.
\end{equation*}
\end{proof}

\subsection{Remarks on diffusion for weakly random Hamiltonian flows}\label{sec:RevisionStochastic}

In this subsection, we recall the main conceptual ideas of references \cite{KP}, \cite{KR1} and \cite{KR2}.

Let $(\ba^\lambda(\overline{t}),\bv^\lambda(\overline{t}))$ be the rescaled positions and velocities of a classical particle in a random Hamiltonian flow defined by 
\begin{align*}
&\ba^\lambda(\overline{t}) : =\lambda^2 \tilde{\ba}(\overline{t}/\lambda^2),\\
&\bv^\lambda(\overline{t}):= \tilde{\bv}(\overline{t}/\lambda^2).
\end{align*}
They satisfy the differential equations
\begin{align*}
&\partial_{\overline{t}} \ba^\lambda = \bv^\lambda,\\
&\partial_{\overline{t}} \bv^\lambda = -\frac{2}{\lambda} \nabla \overline{V}(\ba^\lambda/\lambda^2),
\end{align*}
with initial condition 
\begin{align*}
\ba^\lambda (0) &=0\\
\bv^\lambda (0)&= \bv_0.
\end{align*}
The main difficulty in obtaining the limit of the above stochastic process is that $\ba^\lambda (\overline{t}+\delta\overline{t}), \; \delta\overline{t}\ll 1,$ may be correlated to $\ba^\lambda(\overline{t}),$ and hence the process may be non-Markovian. 

We start with discussing the case $N\ge 3,$ which was studied in \cite{KP}. The authors of \cite{KP} introduce an auxiliary modified dynamics and a stopping time $\tau_\lambda$ such that, up to times $\tau_\lambda,$ the momenta of the modified process are locally aligned and the modified trajectory is a straight line during times of intersection. Furthermore, for times larger than $\tau_\lambda,$ the modified process is a diffusion process. Because of the spatial mixing property of the random potential, the limit of the modified process is Markovian. The proof is completed by showing two more elements: First, the process $(\ba^\lambda (\overline{t}),\bv^\lambda (\overline{t}))_{0< \overline{t}<\tau_\lambda}$ converges weakly to the modified process  as $\lambda\rightarrow 0.$ Second, the stopping time $\tau_\lambda\rightarrow \infty$ as $\lambda\rightarrow 0.$

The case of $N=2$ is a little bit more difficult, due to the fact that the limiting process self-intersects. This difficulty is overcome in \cite{KR1} by modifying the stopping time condition for the modified process: up to stopping time $\tau_\lambda,$ only transversal intersections of the modified trajectory are allowed.

We also recall the main ingredients in \cite{KR2} for proving convergence of the trajectory  of a classical particle in a strongly mixing random potential in dimensions $N\ge 3$ to a spatial Brownian motion in the weak-coupling limit. The proof is based on extending the analysis in \cite{KP} to obtain explicit estimates on the convergence to momentum diffusion, and then using asymptotic expansion to obtain spatial diffusion from momentum diffusion on longer time scales. 

Consider $\overline{\phi}\in C_b^{1,1,2}(\bbR^N\times [0,+\infty)\times \bbR^N\backslash\{0\})$ satisfying 
\begin{equation}
\label{eq:Diff3}
\partial_t \overline{\phi} = \sum_{i,j=1}^N\partial_{k_i} D_{ij}(\bk) \partial_{k_j}\overline{\phi} + \bk\cdot \nabla_\bx \overline{\phi},
\end{equation}
with initial condition $\overline{\phi}(\bx,0,\bk)=\phi_0(\bx,\bk).$ Using an extension of the analysis in \cite{KP} and \cite{BKR}, it is shown  in \cite{KR2} that there exists $\tilde{\beta}>0$ such that, for all compact sets $K\in {\mathcal A}(M),$
\begin{equation}
\label{eq:EstLimit1}
\sup_{(t,\bx,\bk)\in [0,T ]\times K}|{\mathbb E}[\phi^{\lambda}(\bx/\lambda^2,t/\lambda^2,\bk)]-\overline{\phi}(\bx,t,\bk)|\le C T  (1 + \|\phi_0\|_{L^1_\bx L^4_{\bk}})\lambda^{\tilde{\beta}},
\end{equation}
where $C$ is independent of $\lambda$ and $T .$ Note that (\ref{eq:Diff3}) is well-posed in $C_b^{1,1,2},$ see for example \cite{GS1}. Furthermore, the nonvanishing property of the correlation function implies that the diffusion matrix has rank $N-1,$ and hence (\ref{eq:Diff3})  corresponds to diffusion on a sphere of constant momentum, see \cite{KP} and \cite{KR2}. By applying standard asymptotic expansion techniques, one can show that the long-time limit of the solution of (\ref{eq:Diff3}) is spatial diffusion: For every $0<t_0<T <\infty,$ $\overline{\phi}(\bx/\gamma, t/\gamma^2,\bk)$ converges in $C([t_0,T ]; L^\infty (\bbR^N\times \bbR^N))$ as $\gamma\rightarrow 0$ to $u(\bx,t,\bk),$ where $u$ satisfies (\ref{eq:SpaceDiffusion}), with 
\begin{equation}
\label{eq:EstLimit2}
\sup_{t\in[t_0,T ], (\bx,\bk)\in\bbR^N\times\bbR^N}|u(\bx,t,\bk) - \overline{\phi}(\bx/\gamma, t/\gamma^2,\bk)| \le C (\gamma T  + \sqrt{\gamma})\|\phi_0\|_{L^1_\bx L^1_\bk},
\end{equation}
where $C$ is a constant that is independent of $\lambda$ and $T .$ Now, (\ref{eq:EstLimit1}) and (\ref{eq:EstLimit2}) yield (\ref{eq:ClassicalLimit2}).

%%%%%%

%%%%%%%%%%%%%%%%%%%%%%%%%%%%%%%%%%%%%%%%%%%%%%%%%%%%
%%%%%%%%%%%%%%%%%%%%%%%%%%%%%%%%%%%%%%%%%%%%%%%%%%%%

\section{Appendix A: Well-posedness of a generalized nonautonomous nonlinear Schr\"odinger equation}\label{sec:Well-Posedness}

We now discuss the local and global well-posedness of a generalized nonautonomous NLS equation with random (time-dependent) nonlinearities and potential.  Global 
well-posedness and possible occurence of blow up are addressed in \cite{dBD1} and \cite {dBD2}
for  the NLS equation with power nonlinearities and additive
or multiplicative random potential in the form of white noise.

Consider the problem corresponding to a generalized nonlinear Schr\"odinger equation
\begin{equation}
\label{eq:NANLSE}
i\partial_t \psi   = -\Delta\psi   + g(t,\psi  ; \omega), \; \psi(t=0)=\phi,
\end{equation}
where $\omega\in\Omega$ and $g$ contains both the potential and the nonlinearity. Here, $g$ also depends on $\bx\in\bbR^N,$ but we drop the explicit dependence when there is no danger of confusion. 

We say that $(q,r)$ is an admissible pair if 
\begin{align}
r &\in [2,\frac{2N}{N-2}), \; (r\in [ 2,\infty], N=1) \nonumber \\
\frac{2}{q}&=N(\frac{1}{2}-\frac{1}{r}) \label{eq:Admissible}
\end{align}
We make the following assumptions on $g.$

\begin{itemize}

\item[(C1)] The nonlinearity $g=g_1+\cdots +g_k$ with 
$$g_j\in C(\bbR, C(H^1,H^{-1})) \; {\mathbb P}-\mathrm{a.s.}, \; j=1,\cdots,k. $$

\item[(C2)] There exist admissible pairs $(q_j,r_j),$ $j=1,\cdots k,$ such that, for every $T,M>0,$ there exist a constant $C(M)$ independent of $T,$ and $\beta$ independent of $T$ and $M,$ such that
\begin{equation*}
\|g_j(t,u)-g_j(t,v)\|_{L^{r_j'}(\bbR^N)} \le C(M)(1+T^{\beta})\|u-v\|_{L^{r_j}(\bbR^N)} ,\; {\mathbb P}-\mathrm{a.s.},
\end{equation*}
for all $u,v\in H^1$ with $\|u\|_{H^1}+\|v\|_{H^1}\le M,$ and $|t|<T,$ where $r'$ is the conjugate of $r, \ie , 1/r+1/r'=1.$ Furthermore, 
\begin{equation*}
\|g_j(t,u)\|_{W^{1,r_j'}} \le C(M) (1+T^\beta)  (1+\|u\|_{W^{1,r_j}}), \; {\mathbb P}-\mathrm{a.s.},
\end{equation*}
for all $u\in H^1(\bbR^N)\cap W^{1,r}(\bbR^N)$ such that $\|u\|_{H^1}\le M$ and $|t|\le T.$

\item[(C3)] $$\Im g_j(t,u)\overline{u}=0, j=1,\cdots , k,\; {\mathbb P}-\mathrm{a.s.}$$ almost everywhere on $\bbR^N,$ for all $t\in \bbR$ and $u\in H^1.$

\item[(C4)] There exists a functional $G_j\in C(\bbR, C^1(H^1,\bbR))$ with $G_j'=g_j,$ where the prime stands for the Fr\'{e}chet derivative.  We let $G=G_1+\cdots G_k.$ For $u\in H^1,$  
\begin{equation}
\label{eq:GWP2}
|\partial_t G(t,u; \omega)|\le \overline{C}(\|u\|_{L^2})l(t)  \; {\mathbb P}-\mathrm{a.s.},
\end{equation}
where $\overline{C}$ depends only on $\|u\|_{L^2}$ and the real function $l\in L^{\infty}(\bbR)$ such that $l(t)\le 1$ for almost all $t\in \bbR.$

\item[(C5)] For all $M>0,$ there exists $C(M)>0$ and $\epsilon\in (0,1),$ both independent of $t\in \bbR,$ such that 
\begin{equation}
\label{eq:GWP1}
|G(t,u;\omega)|\le \frac{1-\epsilon}{2}\|u\|_{H^1}^2 +C(M), \;{\mathbb P}-\mathrm{a.s.},
\end{equation}
uniformly in $t\in\bbR, \forall u\in H^1,$ such that $\|u\|_{L^2}\le M.$ 
\end{itemize}

%%%%%%%%%%%%%%%%%%%%%%%%%%%%%%%%%%
In what follows, we let $\overline{\Omega}\subset \Omega,$ with $\mu(\overline{\Omega})=1,$ denote the set over which (C1)-(C5) hold.  
We have the following result about local well-posedness almost surely in $H^1.$
 
\begin{proposition}\label{pr:LWP}
Suppose $g$ satisfies assumptions (C1)-(C3). Then the following holds.
\begin{itemize}
\item[(i)] For every $\phi\in H^1(\bbR^N),$ there exists ${\mathbb P}$-a.s. strong $H^1$-solution $u  $ of (\ref{eq:NANLSE}), which is defined on a maximal time interval $(-T_*,T^*),$ such that there exists a blow-up alternative, $\ie, $ if $T^*<\infty,$ $\esssup_{\omega\in\Omega}\|u  (t)\|_{H^1}\rightarrow \infty$ as $t\nearrow T^*,$ and if $T_*<\infty,$   $\esssup_{\omega\in\Omega}\|u  (t)\|_{H^1}\rightarrow \infty$ as $t\searrow -T_*.$ Moreover,
\begin{equation*}
u  \in L^a_{loc}((-T_*,T^*), W^{1,b}(\bbR^N)),\ \ \omega\in\overline{\Omega}
\end{equation*}
for all admissible pairs $(a,b).$ 
\item[(ii)] The charge is ${\mathbb P}$-a.s. conserved,
\begin{equation*}
\|u  (t)\|_{L^2} = \|\phi\|_{L^2}, \ \ \omega\in\overline{\Omega},
\end{equation*}
for all $t\in (-T_*,T^*).$ 
\item[(iii)] ${\mathbb P}$-a.s., $u$ depends continuously on $\phi \; :$ If $\phi_n\stackrel{n\rightarrow\infty}{\rightarrow} \phi$ in $H^1,$ and if $u_{ n}$ is the maximal solution of (\ref{eq:NANLSE}) corresponding to the initial condition $\phi_n,$ then $u_{ n}\stackrel{n\rightarrow\infty}{\rightarrow} u  $ in $C([-S,T], L^p(\bbR^N))$ for every compact interval $[-S,T]\subset (-T_*,T^*)$ and $p\in [2,\frac{2N}{N-2}) \; (p\in [2,\infty) , N=1).$ 

\end{itemize}
\end{proposition}

The proof of the above proposition is a direct extension of the deterministic case with $\omega\in\overline{\Omega}$; see the Appendix in \cite{A-S1} for a discussion of the latter. It is based on Kato's method, which relies on Strichartz estimates and a fixed point argument, \cite{Ka1,Ka2}.

Proving global well-posedness for data which are not necessarily small is a little bit more delicate, since energy is not conserved.
We now define the energy functional
\begin{equation}
\label{eq:Energy}
E(t,u;\omega):= \frac{1}{2}\int |\nabla u|^2 dx +G(t,u;\omega),
\end{equation}
for $u\in H^1(\bbR^N)$ and $\omega\in\Omega.$ Since the nonlinearity and the potential depend on time, the energy is not conserved. We have the following proposition.

\begin{proposition}\label{pr:EnergyUpperBd}
Suppose that (C1)-(C4) are satisfied, and let $u  $ denote the solution of (\ref{eq:NANLSE}) given by Proposition \ref{pr:LWP}. Then 
\begin{equation}
\label{eq:EnergyUpperBd}
|E(t,u  (t);\omega)| \le |E(0,\phi;\omega)| + T \overline{C}(\|\phi\|_{L^2}), \ \ {\mathbb P}-\mathrm{a.s.},
\end{equation}
for all $t\in [-T,T],$ where $[-T,T]$ is a compact subset of $(-T_*,T^*),$ and $\overline{C}(\|\phi\|_{L^2})$ appears in (C4). 
\end{proposition}

\begin{proof}
Let $\omega\in\overline{\Omega}.$ Since (C1)-(C3) are satisfied, the results of Proposition \ref{pr:LWP} hold. We choose a finite $T>0$ such that $T<\min(T_*,T^*).$ We have from Proposition \ref{pr:LWP} that 
\begin{equation*}
u   \in L^a_{loc}((-T_*,T^*), W^{1,b}(\bbR^N)), \ \ \omega\in\overline{\Omega}
\end{equation*}
for all admissible pairs $(a,b).$ In particular, 
\begin{equation*}
u  \in L^{q_j}((-T,T), W^{1,r_j}(\bbR^N)), \; j=1,\cdots ,k, \ \ \omega\in\overline{\Omega}
\end{equation*}
where the admissible pairs $(q_j,r_j)$ appear in (C2). 

Since $\nabla$ commutes with the unitary propagator $U$ the $L^2$ norm is invariant under unitary transformations, we have using the Duhamel formula that
\begin{align*}
&\|\nabla u  (t)\|^2_{L^2} = \|\nabla U(0,t) u  (t) \|^2_{L^2} \\
&= \|\nabla \phi - i \int_0^t ds~ U(0,s) \nabla g(s,u  (s);\omega)\|^2_{L^2} \\
&= \|\nabla \phi\|_{L^2}^2 -2 \Im \langle \nabla \phi ,\int_0^tds~ U(0,s) \nabla g(s,u  (s);\omega) \rangle + \| \int_0^t ds~ U(0,s)\nabla g(s,u  (s);\omega) \|_{L^2}^2 \\
&= \|\nabla \phi\|_{L^2}^2 +2\Im \int_0^t ds~ \langle \nabla g(s,u  (s);\omega), U(s,0)\nabla \phi \rangle + \| \int_0^t ds~ U(0,s)\nabla g(s,u  (s);\omega) \|_{L^2}^2\\
&= \|\nabla\phi\|_{L^2}^2 + 2 \sum_{j=1}^k \Im \int_0^t ds~\langle \nabla g_j(s,u  (s);\omega) , \nabla u  (s)\rangle \\
&= \|\nabla \phi\|_{L^2}^2 - 2 \sum_{j=1}^k \Im \int_0^t ds~ \langle g_j(s,u  (s);\omega), \Delta u  (s)\rangle,\ \ \omega\in\overline{\Omega}
\end{align*}
where the scalar product is well-defined using (C2) and duality on 
\begin{equation*}
(L^1((-T,T),H^1)+L^{q'_j}((-T,T),H^{1,r_j'}))\times (L^\infty((-T,T),H^1)\cap L^{q_j}((-T,T),H^{1,r_j}))), 
\end{equation*}
$j=1,\cdots ,k,$ see for example \cite{BerLoef1}. Here, 
\begin{equation*}
H^{s,p}(\bbR^N):= \{u\in {\mathcal S}'(\bbR^N): \cF^{-1} (1+|k|^2)^{\frac{s}{2}}\cF u \in L^p(\bbR^N)\}, \; s\in\bbR, 1\le p\le \infty,
\end{equation*}
where ${\mathcal F}$ stands for the Fourier transform; and the space $H^{s,p}$ is equipped with the norm
\begin{equation*}
\|u\|_{H^{s,p}} = \|{\mathcal F}^{-1}(1+\|\bk\|^2)^{\frac{s}{2}} {\mathcal F} u\|_{L^p} , \; u\in H^{s,p}(\bbR^N).
\end{equation*}

Now,
\begin{align*}
\Im\langle g(t,u  (t);\omega) , \Delta u  (t) \rangle &= \lim_{\epsilon\searrow 0} \Im \langle (1-\epsilon\Delta)^{-1} g(t,u  (t);\omega) , (1-\epsilon \Delta)^{-1} \Delta u  (t)\rangle \\
&= \lim_{\epsilon\searrow 0} \Im \langle (1-\epsilon\Delta)^{-1} g(t,u(t)), (1-\epsilon\Delta)^{-1} (-i\partial_t u(t) + g(t,u(t))) \rangle \\
&= \Re\langle g(t,u  (t);\omega), \partial_t u  (t)\rangle \\
&=  \frac{d}{dt} G(t,u  (t);\omega) - (\partial_t G)(t,u  (t);\omega), \ \ \omega\in\overline{\Omega}
\end{align*}
for almost all $t\in (-T,T),$where $G$ appears in Assumption (C4). Therefore,
\begin{equation}
\label{eq:RateH1}
\|\nabla u  (t)\|_{L^2}^2 =\|\nabla \phi\|_{L^2}^2 - 2 G(t,u  (t);\omega)+2 G(0,\phi;\omega) +2\int_0^t ds~ (\partial_sG)(s,u  (s);\omega),
\end{equation}
for $\omega\in\overline{\Omega}.$
Together with (\ref{eq:Energy}), Assumption (C4) and conservation of charge, this implies
\begin{equation*}
|E(t,u  (t);\omega)| \le |E(0,\phi;\omega)| + T \overline{C}(\|\phi\|_{L^2}), 
\end{equation*}
for $t\in [-T,T]$ and all $\omega\in\overline{\Omega}.$   
\end{proof}

We have the following theorem about ${\mathbb P}$-a.s. global well-posedness in $H^1.$

\begin{theorem}\label{th:GWP}
Suppose (C1)-(C5) hold. Then the solution $u $ of (\ref{eq:NANLSE}) with initial condition $\phi\in H^1 (\bbR^N )$ is ${\mathbb P}$-a.s.  global in $H^1,$ i.e., for $\omega\in\overline{\Omega},$
$$T^*=T_*=\infty ,$$ where $T^*,T_*$ appear in Proposition \ref{pr:LWP}. 
\end{theorem}

\begin{proof}
Let $\omega\in\overline{\Omega}.$ Propositions \ref{pr:LWP} and \ref{pr:EnergyUpperBd} follow from (C1)-(C4). It suffices to show that $\|u  (t)\|_{H^1}, \; t\in [0,T^*)$ is finite if $T^*<\infty,$ which, together with the blow-up alternative, implies a contradiction. Suppose $T^*<\infty.$ 
Assumption (C5), charge conservation and (\ref{eq:EnergyUpperBd}) imply that
\begin{align*}
&\frac{1}{2}\|u  (s)\|^2_{H^1} = \frac{1}{2}(\|u  (s)\|^2_{L^2}+\|\nabla u  (s)\|^2_{L^2})\\
& \le \frac{1}{2}\|u  (s)\|^2_{L^2} + |E(s,u  (s);\omega)| + |G(s,u  (s);\omega)| \\
& < \frac{1}{2} \|\phi\|^2_{L^2} + |E(0,\phi;\omega)| + T^*\overline{C}(\|\phi\|_{L^2}) + \frac{1-\epsilon}{2}\|u  (s)\|^2_{H^1} + C(\|\phi\|_{L^2}) , 
\end{align*}
for all $s\in [0,T^*),$ $\omega\in\overline{\Omega},$
and hence 
\begin{equation*}
\sup_{\omega\in\overline{\Omega}}\sup_{s\in [0,T^*)}\|u  (s)\|_{H^1} <\infty,  
\end{equation*}
for finite $T^*,$ which contradicts the blow-up alternative. The case of $T_*$ is proven similarly. 

\end{proof}

\section{Appendix B: Rate of change of field energy and momenta}\label{app:Ehrenfest}

\begin{proof}[Proof of (\ref{eq:RatePotential}), Section \ref{sec:ProofMain3}] Differentiating $\langle \psi, \lambda V_h \psi\rangle$ with respect to $t$ and using (\ref{eq:NLSE}), Assumption (A1), and the fact that $\psi\in H^1(\bbR^N),$ we have, for $\omega\in\overline{\Omega},$ 
\begin{align*}
\partial_t \langle \psi, \lambda V_h \psi\rangle &= \langle \psi, \lambda\partial_t V_h \psi \rangle + \langle \partial_t \psi, \lambda V_h \psi \rangle  + \langle \psi,\lambda V_h \partial_t\psi\rangle \\ 
&= \langle \psi,\lambda \partial_t V_h \psi \rangle + \langle -\Delta \psi + \lambda V_h \psi -f(\psi), i\lambda V_h\psi\rangle \\ &+ \langle i\psi , \lambda V_h (-\Delta \psi +\lambda V_h\psi -f(\psi))\rangle \\
&= \langle \psi,\lambda \partial_t V_h \psi \rangle + 2\langle i\lambda\nabla V_h\psi, \nabla\psi\rangle ,
\end{align*}
where we have used integration by parts in the last step. 
\end{proof}

\begin{proof}[Proof of (\ref{eq:RateEnergy}), Section \ref{sec:ProofMain3}] The proof of (\ref{eq:RateEnergy}) follows directly from (\ref{eq:RateH1}) in the proof of Proposition \ref{pr:EnergyUpperBd}, Appendix A, with the identification
\begin{equation*}
g(t,u) = \lambda V_h(t)u - f(u).
\end{equation*}
\end{proof}

\begin{proof}[Proof of (\ref{eq:Ehrenfest}), Section \ref{sec:ProofMain3}] We use a regularization scheme similar to the one used in Proposition \ref{pr:EnergyUpperBd}, Appendix A. Let $I_\epsilon := (1-\epsilon\Delta)^{-1}.$ For all $\omega\in\overline{\Omega},$ we have that 
\begin{align*}
&\partial_t \langle i\psi ,\nabla \psi \rangle = \partial_t \lim_{\epsilon\searrow 0} \langle I_\epsilon i \psi, I_\epsilon \nabla \psi \rangle \\
&= \lim_{\epsilon\searrow 0}\{ \langle I_\epsilon i\partial_t \psi , I_\epsilon \nabla\psi \rangle + \langle I_\epsilon i\psi, I_\epsilon \nabla \partial_t \psi\rangle \} \\
&= \lim_{\epsilon \searrow 0} \{ \langle I_\epsilon (-\Delta \psi + \lambda V_h\psi -f(\psi)), I_\epsilon \nabla \psi\rangle  - \langle I_\epsilon \psi, I_\epsilon \nabla (-\Delta \psi + \lambda V_h\psi -f(\psi)) \rangle \}\\
&=\lambda \langle \psi, V_h \nabla \psi\rangle - \lambda\langle \psi, \nabla (V_h\psi)\rangle \\
&= - \lambda \langle \psi, \nabla V_h\psi\rangle .
\end{align*}
\end{proof}

%%%%%%%%%%%%%%%%%%%%%%%%%%%%%%%%%%%%%%%%%%%%%%%%%%%%%%
%BIBLIOGRAPHY %%%%%%%%%%%%%%%%%%%%%%%%%%%%%%%%%%%%%%%%
%%%%%%%%%%%%%%%%%%%%%%%%%%%%%%%%%%%%%%%%%%%%%%%%%%%%%%

\end{document}